\documentclass[pra,onecolumn,superscriptaddress,aps,5pt,showpacs,showkeys]{revtex4}
\usepackage{graphicx}  
\usepackage{dcolumn}   
\usepackage{bm}        
\usepackage{amsfonts}
\usepackage{amssymb,latexsym,amsmath}
\usepackage{amsthm}
\usepackage{float}
\usepackage{times}
\usepackage{subfigure}
\usepackage[usenames]{color}
\usepackage{ulem}
\usepackage{epstopdf}
\usepackage{color}

\makeatletter
\let\save@mathaccent\mathaccent
\newcommand*\if@single[3]{%
  \setbox0\hbox{${\mathaccent"0362{#1}}^H$}%
  \setbox2\hbox{${\mathaccent"0362{\kern0pt#1}}^H$}%
  \ifdim\ht0=\ht2 #3\else #2\fi
  }
\newcommand*\rel@kern[1]{\kern#1\dimexpr\macc@kerna}
\newcommand*\widebar[1]{\@ifnextchar^{{\wide@bar{#1}{0}}}{\wide@bar{#1}{1}}}
\newcommand*\wide@bar[2]{\if@single{#1}{\wide@bar@{#1}{#2}{1}}{\wide@bar@{#1}{#2}{2}}}
\newcommand*\wide@bar@[3]{%
  \begingroup
  \def\mathaccent##1##2{%
    \let\mathaccent\save@mathaccent
    \if#32 \let\macc@nucleus\first@char \fi
    \setbox\z@\hbox{$\macc@style{\macc@nucleus}_{}$}%
    \setbox\tw@\hbox{$\macc@style{\macc@nucleus}{}_{}$}%
    \dimen@\wd\tw@
    \advance\dimen@-\wd\z@
    \divide\dimen@ 3
    \@tempdima\wd\tw@
    \advance\@tempdima-\scriptspace
    \divide\@tempdima 10
    \advance\dimen@-\@tempdima
    \ifdim\dimen@>\z@ \dimen@0pt\fi
    \rel@kern{0.6}\kern-\dimen@
    \if#31
      \overline{\rel@kern{-0.6}\kern\dimen@\macc@nucleus\rel@kern{0.4}\kern\dimen@}%
      \advance\dimen@0.4\dimexpr\macc@kerna
      \let\final@kern#2%
      \ifdim\dimen@<\z@ \let\final@kern1\fi
      \if\final@kern1 \kern-\dimen@\fi
    \else
      \overline{\rel@kern{-0.6}\kern\dimen@#1}%
    \fi
  }%
  \macc@depth\@ne
  \let\math@bgroup\@empty \let\math@egroup\macc@set@skewchar
  \mathsurround\z@ \frozen@everymath{\mathgroup\macc@group\relax}%
  \macc@set@skewchar\relax
  \let\mathaccentV\macc@nested@a
  \if#31
    \macc@nested@a\relax111{#1}%
  \else
    \def\gobble@till@marker##1\endmarker{}%
    \futurelet\first@char\gobble@till@marker#1\endmarker
    \ifcat\noexpand\first@char A\else
      \def\first@char{}%
    \fi
    \macc@nested@a\relax111{\first@char}%
  \fi
  \endgroup
}
\makeatother

\newtheorem{lemma}{Lemma}
\newtheorem{proposition}{Proposition}

\newtheorem{remark}{Remark}

\begin{document}



\title{Linear and nonlinear dynamics of isospectral granular chains}

\author{R. Chaunsali}
\affiliation{Aeronautics and Astronautics, University of Washington, Seattle, WA 98195-2400}
\author{H. Xu \footnote{The first two authors contributed equally.}}
\affiliation{Department of Mathematics and Statistics, University of Massachusetts, Amherst, MA 01003-4515}
\author{J. Yang}
\affiliation{Aeronautics and Astronautics, University of Washington, Seattle, WA 98195-2400}
\author{P. G. Kevrekidis \footnote{kevrekid@gmail.com}}
\affiliation{Department of Mathematics and Statistics, University of Massachusetts, Amherst, MA 01003-4515}




\date{\today}

\begin{abstract}
We study the dynamics of isospectral 
granular chains that are highly tunable due to the nonlinear Hertz contact law interaction between the granular particles. The system dynamics can thus be tuned easily from being linear to strongly nonlinear by adjusting the initial compression applied to the chain. In particular, we introduce both discrete and continuous spectral transformation schemes to generate a family of granular chains that are isospectral in their linear limit. Inspired by the principle of supersymmetry in quantum systems, we also introduce a methodology to add or remove certain eigenfrequencies, and we demonstrate numerically that the corresponding physical system can be constructed in the setting of one-dimensional granular crystals. In the linear regime, we highlight the similarities in the elastic wave transmission characteristics of such isospectral systems, and emphasize that the presented mathematical framework allows one to suitably tailor the wave transmission through a general class of granular chains, both ordered and disordered. Moreover, we show how the dynamic response of these structures deviates from its linear limit as we introduce Hertzian nonlinearity in the chain and how nonlinearity 
breaks the notion of linear isospectrality. 

\end{abstract}

\maketitle
\section{Introduction}
The science of manipulating stress waves propagating through discrete media has attracted significant research interest in recent years. A granular crystal can be such a medium, which is a periodic structure composed of discrete particles that interact elastically as per Hertz's contact law \cite{physicstoday, Nesterenko, Sen}. At a fundamental level, the periodicity in granular crystals causes Bragg scattering, which allows them to filter a range of elastic waves, thereby forming frequency band structures \cite{Brillouin}. Furthermore, such structures are special in the sense that the system dynamics can be tuned from linear to highly nonlinear regimes, simply by changing the initial static compression applied externally~\cite{Nesterenko}. Thus, due to the tunable interplay of nonlinearity and discreteness, these systems have become popular testbeds for demonstrating a wide range of nonlinear wave phenomena, for example, solitary waves \cite{physicstoday,Nesterenko, Coste, Sen, Jayaprakash}, shock waves \cite{Hong, Daraio, Doney}, discrete breathers \cite{physicstoday,Boechler, Chong}, energy cascades \cite{Kim1}, and nanoptera \cite{Kim2,atanas}.

By leveraging this remarkable tunability, we direct our efforts to study isospectral systems based on granular media. Isospectral systems refer to a family of systems with identical spectra. In mechanical settings, these isospectral systems would represent a group of structures with identical eigenfrequencies. The existence of isospectral systems is closely linked to a non-uniqueness problem, which was articulately posed as ``Can we hear the shape of a drum?'' in the famous paper of M. Kac~\cite{Kac}. The non-uniqueness lies in the possibility of constructing multiple structures (shapes) from a given set of eigenfrequencies (hearing) \cite{Gordon}. In the past, this avenue has been primarily investigated as an inverse problem \cite{Gladwell, Gottlieb, Giraud}.  It is only recent that such systems, in conjunction with the systems inspired by a closely related principle of supersymmetry \cite{Cooper}, have gained renewed interest in a broad spectrum of research communities, especially in the context of designing strategic architectures to control wave transmission characteristics. For example, transparent interfaces \cite{Longhi1}, mode converters \cite{Heinrich}, reflection-less bent waveguides \cite{Campo}, and supersymmetric optical structures~\cite{extras1} have been proposed. These concepts have also driven fundamental studies on Bloch waves in defective \cite{Longhi2} and disordered \cite{Yu} lattices. However, wave transmission characteristics in mechanical isospectral systems have not been
systematically explored. Moreover, the effect of nonlinearity in the dynamics of initially isospectral systems is also a fundamental question to be answered.

Owing to the fact that a tunable granular system can be used as a fertile testbed to explore both of these issues, our aim in the present work is to harness the wealth of mathematical schemes, such as discrete and continuous spectral transformations, to generate a family of isospectral granular chains. In the linear regime, we evaluate and compare elastic wave transmission characteristics of these systems. Inspired by the principle of supersymmetry, we extend our mathematical framework to construct a class of granular chains that have specific eigenfrequencies added or removed.  In this way, we show that the entire framework gives us remarkable freedom in constructing a general class of granular chains with desired wave transmission characteristics in the linear regime. 
In the nonlinear regime, we focus on studying the versatile dynamics of isospectral granular chains, and discuss how the dynamics deviates from the linear expectation of isospectrality with the extent of nonlinearity in the medium.

This manuscript is structured as follows: in Section II, we present a mathematical model of granular systems and formulate their equations of motion. We then linearize the system, which we use for constructing isospectral chains. In Section III, we explain spectral transformation schemes to get isospectral spring-mass systems. In Section IV, we discuss the possibility of constructing a family of isospectral systems in the setting of granular chains, such that their linearized systems share the same eigenfrequencies. Additionally, we introduce a strategy for removing or adding eigenfrequencies from the linearized system of a granular chain. In Section V,  both linear and nonlinear dynamics of the family of such granular chains is studied and compared, while in section VI we summarize our findings and present our conclusions, as well as a number of directions for future study.

\section{Theoretical setup}
In this work we are interested in granular chains consisting of spherical beads that interact through the Hertzian contact law (i.e., $F \sim d^{3/2}$, where $F$ and $d$ are the beads' compressive force and displacement, respectively)~\cite{Johnson}. In particular, we consider a granular chain, denoted by (NS-1), that has $n$ spherical beads. Suppose the chain is bounded by half-spaces from both ends, and the chain starts with precompression between contacting beads and between the end bead and the half-space as well. For simplicity, we assume that all beads and the half-spaces are made from the same material, so that they share the same density $\rho$, Young's modulus $E$, and Poisson's ratio $\nu$. However, the beads can have different radii $(r_1,r_2,...,r_n)$, which will lead to different masses $(m_1,m_2,...,m_n)$ and different inter-particle contact stiffness coefficients $(k_1,k_2,...,k_{n+1})$, as shown below:
\begin{eqnarray}
m_j &=& \frac{4}{3}\pi\rho r_j^3, \quad 1\leq j \leq n \\
k_1 &=& \frac{4}{3}E_*\sqrt{r_1}, \\
k_j &=& \frac{4}{3}E_*\sqrt{\frac{r_{j-1}r_j}{r_{j-1}+r_j}}, \quad 2\leq j \leq n \\
k_{n+1} &=& \frac{4}{3}E_*\sqrt{r_n},
\end{eqnarray}

where $E_*$ is $\frac{2E}{3(1-\nu^2)}$. Note that $m_j$'s and $k_j$'s are solely functions of the radii, provided that the material of the beads do not change.

Then the equations of motion for the granular chain can be written as:
\begin{equation}
\label{eqn0_gran_1}
m_j \ddot{z}_j=k_j(d_j+z_{j-1}-z_j)_+^p-k_{j+1}(d_{j+1}+z_j-z_{j+1})_+^p, \quad 1\leq j \leq n\\
\end{equation}

where $z_0=z_{n+1}=0$, $p=\frac{3}{2}$ due to the Hertzian contact, $d_j$ denotes the precompression between $(j-1)$-th and $j$-th beads, $z_j$ is the displacement of the $j$-th bead, and $(a)_+:=\max(a,0)$.



If $z_1=z_2=...=z_n=0$ is a steady-state solution to the system (i.e., $z_j$ is measured with respect to its equilibrium position), the precompression should satisfy
$$k_1 d_1^p=k_2 d_2^p=...=k_{n+1} d_{n+1}^p=C,$$
where $C$ is a constant controlling the static precompression applied to the granular chain. To study the linear limit of the system,
we assume the displacements $z_j \ll d_j$ and linearize the equations~(\ref{eqn0_gran_1}) as:
\begin{equation}
\label{eqn0_ln_1}
m_j\ddot{z}_j=K_{j}(z_{j-1}-z_j)-K_{j+1}(z_j-z_{j+1}), \quad 1\leq j\leq n
\end{equation}
where $z_0=z_{n+1}=0$ and $K_j=p d_j^{p-1} k_j=p k_j^{1/p} C^{(p-1)/p}$ for $1\leq j\leq n+1$.

We notice that this set of equations actually describes a (linear) spring-mass system, which we denote as (S-1), with masses $(m_1,m_2,...,m_n)$ and springs $(K_1,K_2,...,K_{n+1})$. The equations~(\ref{eqn0_ln_1}) can also be written using matrices as
\begin{equation}
\label{eqn0_ln_matrix}
M \ddot{\boldsymbol{z}}+B \boldsymbol{z}=0
\end{equation}
where $M=\left(
    \begin{array}{ccccc}
        m_1 & 0 & 0 & ... & 0\\
        0 & m_2 & 0 & ... & 0\\
        ... & ... & ... & ... & ...\\
        0 & 0 & ... & 0 & m_n
    \end{array}
\right)$ and
$B=\left(
    \begin{array}{ccccc}
        K_1+K_2 & -K_2 & 0 & ... & 0\\
        -K_2 & K_2+K_3 & -K_3 & ... & 0\\
        ... & ... & ... & ... & ...\\
        0 & 0 & ... & -K_n & K_n+K_{n+1}
    \end{array}
\right)$.\\


\noindent It is known that isospectral spring-mass systems can be constructed such that they bear the same natural frequencies as those in (S-1)~\cite{Gladwell}. In the coming Section III, we briefly review this approach and upgrade it by introducing a continuous isospectral transformation. Then we take a further step to tailor and apply the enhanced method for constructing isospectral granular chains, which is described in Section IV.

\section{Isospectral spring-mass system}
\label{section_isospectral_linear}

To find the eigenfrequencies of the spring-mass system (S-1), which is described by~(\ref{eqn0_ln_1}),  we set $z_j=Z_j e^{i\omega t}$ and $\lambda=\omega^2$ in equation~(\ref{eqn0_ln_matrix}) and solve the eigenvalue system $$(B-\lambda M)\boldsymbol{Z}=0.$$
Let $G=M^{1/2}$ and $H=G^{-1}B G^{-1}$, then the eigenvalue problem yields $$(H-\lambda I)(G\boldsymbol{Z})=0,$$
where the eigenvalues $\lambda$ of $H$ give the eigenfrequencies ($\omega=\pm\sqrt{\lambda}$) of the spring-mass system. Here we call $H$ the associated matrix of the spring-mass system (S-1), and the following remarks explain the relationships between a spring-mass system and its associated matrix.

\begin{remark}
\label{remark_mass_spring_1}
If a spring-mass system has positive masses and spring constants, i.e., $m_i>0$ for $1\leq i\leq n$ and $K_j>0$ for $1\leq j\leq n+1$, then its associated matrix $H$ satisfies the following conditions:
\begin{itemize}
\item[(A1)] $H$ is a real symmetric tridiagonal matrix;
\item[(A2)] all entries on the main diagonal are real, and all off-diagonal entries are negative;
\item[(A3)] $H$ is positive definite. (See Appendix for proof)
\end{itemize}
In particular, for the linearized system of a granular chain system, its associated matrix $H$ satisfies all aforementioned conditions.
\end{remark}

\begin{remark}
\label{remark_mass_spring_2}
If $\tilde{H}$ is a $n\times n$ matrix and it satisfies conditions (A1)--(A3), then there exists a family of spring-mass systems with positive parameters $(\alpha,\beta)$ such that
\begin{itemize}
\item each spring-mass system in this family has positive masses $(\tilde{m}_1,\tilde{m}_2,...,\tilde{m}_n)$ and spring constants $(\tilde{K}_1,\tilde{K}_2,...,\tilde{K}_{n+1})$;
\item the associated matrix of each spring-mass system in this family is $\tilde{H}$.
\end{itemize}
In particular, if we additionally require that $\frac{\tilde{K}_1}{\sqrt{\tilde{m}_1}}$ and $\frac{\tilde{K}_{n+1}}{\sqrt{\tilde{m}_n}}$ are fixed in the spring-mass system, then there is only one spring-mass system in the family and it satisfies $(\alpha,\beta)={{(\frac{\tilde{K}_1}{\sqrt{\tilde{m}_1}}},} \frac{\tilde{K}_{n+1}}{\sqrt{\tilde{m}_n}})$. (See Appendix for proof)
\end{remark}

\begin{remark}
\label{remark_mass_spring_3}
Let $T_{A,n}$ be the space of $n\times n$ matrices that satisfy conditions (A1)--(A3). Let $L_{\alpha,\beta,n}$ denote the space of spring-mass systems with masses and springs satisfying $\frac{K_1}{\sqrt{m_1}}=\alpha$ and $\frac{K_{n+1}}{\sqrt{m_n}}=\beta$.
If we define a function by mapping a spring-mass system in $L_{\alpha,\beta,n}$ to its associated matrix, then it will be a smooth bijective function connecting $L_{\alpha,\beta,n}$ and $T_{A,n}$.

\end{remark}

As these remarks suggest, it is sufficient to consider the matrices in $T_{A,n}$ for investigating the eigenfrequencies of spring-mass systems. That is to say, we will be able to obtain isospectral spring-mass systems if we can find some transformation that maps within $T_{A,n}$ (i.e., retains the conditions (A1)--(A3)) and keeps the eigenvalues invariant. For isospectral flow that maps the matrix $H$ to $\tilde{H}$, here we consider the following approaches:
\begin{itemize}
\item QR decomposition: $QR=H-\mu I_n$, $\tilde{H}=RQ+\mu I_n$.

\item Cholesky decomposition: $L_1 L_1^T=H-\mu I_n$, $H_1=L_1^T L_1+\mu I_n$, $L_2 L_2^T=H_1-\mu I_n$, $\tilde{H}=L_2^T L_2+\mu I_n$. \\
We perform the process twice here since a single application may not yield
an ``acceptable'' $\tilde{H}$ for some choices of $\mu$. Nevertheless, 
repeating the process a second time can be effectively considered as
a QR decomposition with $Q=L_1 L_2^{-T}$ and $R=L_2^T L_1^T$ and hence
should work for generic values of $\mu$.

%
%


\item Continuous isospectral flow such as Toda flow~\cite{Moser,Deift1}. 
\end{itemize}
While any of these approaches will give us an isospectral transformation (see e.g.~\cite{Gladwell} for details about the approaches using QR decomposition or Cholesky decomposition), we will place special emphasis on continuous isospectral flow in this study.
Despite the generality that the continuous approach offers, there are also practical reasons for this choice:
The first two transformation schemes are discrete in the sense that a family of isospectral systems is not evolved smoothly by varying a parameter in space. That is, though these schemes have been used for constructing isospectral or supersymmetric systems in the past \cite{Gladwell, Heinrich}, in mechanical settings, we observe that these approaches often result in systems with extremely high contrast (variations) in system parameters. This makes it very difficult to practically realize such systems. Therefore, we adopt the continuous isospectral flow, in which we can control one (or more) parameter(s) and generate the isospectral systems without incurring sharp contrasts in system parameters. This gives us extra degrees of freedom in terms of \textit{smoothly} controlling the parameters of the resulting mechanical systems.

Let $T_{\Lambda}$ be the space of all $n\times n$, real, symmetric, tridiagonal matrices with negative subdiagonals and spectrum $\Lambda=\{\lambda_1>\lambda_2>...>\lambda_n>0\}$. Since all eigenvalues in the spectrum are positive, all the matrices in this space are positive definite thus have positive main diagonals, i.e., $T_{\Lambda}\subseteq T_{A,n}$. On the other hand, it has been shown in~\cite{Deift1} that the spectrum of a real, symmetric, tridiagonal matrix with negative subdiagonals must be simple (this implies that any spring-mass system should have distinct eigenfrequencies) thus $T_{A,n}\subseteq \bigcup_{\Lambda}T_{\Lambda}$.
That is to say, $T_{\Lambda}$ is exactly the collection of matrices that bear the spectrum $\Lambda$ in $T_{A,n}$.
In the following Lemma, we review the topology of $T_{\Lambda}$ described in~\cite{Deift1,Deift} which naturally defines a continuous isospectral flow.

\begin{lemma}
\label{lemma_continuous}
The space $T_{\Lambda}$ is diffeomorphic to $S^{n-1}_{>0}:=\{\boldsymbol{x}\in\mathbb{R}^n | \sum_{i=1}^n x_i^2=1, ~x_i>0, 1\leq i\leq n \}$, i.e., there exists a bijective and smooth map $\phi: T_{\Lambda}\to S^{n-1}_{>0}$.
\end{lemma}

\begin{proof}
For $H\in T_{\Lambda}$, there exists a unique $U$ such that $H=U D U^T$ and $U U^T=U^T U=I_n$ where $D:=\mathrm{diag}(\lambda_1,\lambda_2,...,\lambda_n)$ and $U_{1j}>0$ for $1\leq j\leq n$. Define $\phi(H):=(U_{11},U_{12},...,U_{1n})$, then it maps a matrix in $T_{\Lambda}$ to a point on the sphere $S^{n-1}_{>0}$.

On the other hand, given a point $\boldsymbol{x}=(x_1,x_2,...,x_n)$ on the sphere $S^{n-1}_{>0}$, we can find its corresponding matrix $H=\phi^{-1}(\boldsymbol{x})$ in the space $T_{\Lambda}$. To be more specific, suppose $U$ is an orthogonal matrix satisfying $U_{1j}=x_j$ for $1\leq j\leq n$ and there exists some $H$ such that $H=U D U^T$, i.e., $H U=U D$. Then we have $H {\rm col}_j(U)= \lambda_j {\rm col}_j(U)$ for $1\leq j\leq n$ and the orthogonality of columns of $U$ leads to the following equations:
\begin{eqnarray}
a_i&=&\sum_{j=1}^n \lambda_j U_{ij}^2, \quad 1\leq i\leq n \\
b_i&=&-\sqrt{ \sum_{j=1}^n[ (\lambda_j-a_i)U_{ij}-b_{i-1}U_{i-1,j} ]^2 }, \quad 1\leq j\leq n-1 \\
U_{i+1,j}&=& \frac{1}{b_i}[ (\lambda_j-a_i)U_{ij}-b_{i-1}U_{i-1,j} ]  , \quad 1\leq j\leq n-1
\end{eqnarray}
where $b_0=b_n=U_{0j}=0$ for $1\leq j\leq n$. By evaluating the equations in the order
\[
a_1\to b_1\to U_{2,j} \to a_2 \to b_2 \to U_{3,j} \to a_3 \to ... \to U_{n,j}\to a_n ,
\]
we can obtain a unique matrix $H=\phi^{-1}(\boldsymbol{x})$. It can be easily shown that $b_i\neq 0$ for $1\leq i\leq n-1$ thus such $H$ will always exist.

\end{proof}

By Lemma~\ref{lemma_continuous}, the matrices in $T_{\Lambda}$ are parametrized by the coordinates in $S^{n-1}_{>0}$ thus every matrix in such a space will be accessible to us through its image point on the sphere. That is to say, by continuously moving the image point $\boldsymbol{x}$ in $S^{n-1}_{>0}$ we can accordingly obtain a continuous isospectral flow of matrices in $T_{\Lambda}$ by $H=\phi^{-1}(\boldsymbol{x})$. Also utilizing Remark~\ref{remark_mass_spring_3}, we have shown that every spring-mass system will have a whole family of isospectral spring-mass systems with parameters $(x_1,x_2,...,x_n) \in S^{n-1}_{>0}$.

\section{Isospectral granular chains}

\subsection{Calculation of isospectral systems}
Given a granular chain (NS-1) whose linearized system is (S-1), we now consider constructing another granular chain (NS-2) such that its linearized system (S-2) has the same eigenfrequencies as those in (S-1). In general, this is a challenging task to achieve, because it is not always the case that a spring-mass system is the linearization of some granular chain (although we have seen that the reverse is true) and our isospectral transformations for spring-mass systems do not necessarily map a linearized granular chain to another linearized granular chain. To be more specific, a granular chain is determined by $n$ free variables $(r_1,r_2,...,r_n)$ while a spring-mass systems has $(2n+1)$ free variables, namely $(m_1,m_2,...,m_n)$ and $(K_1,K_2,...,K_{n+1})$. Since the symmetric tridiagonal matrices in the isospectral transformations have $(2n-1)$ non-zero elements, only $(2n-1)$ degrees of freedom in isospectral spring-mass systems can be kept, as also seen in Remark~\ref{remark_mass_spring_3}.

To enable the construction of a granular chain from a tridiagonal matrix, we would like to extend our consideration to generalized forms of granular systems (Fig.~\ref{fig1}). In particular, motivated in part by the proposals
of~\cite{guillaume}, 
we attach a linear spring to each bead in the granular chain and let the other end of the spring be fixed. 
Note that in Fig.~\ref{fig1}, these linear springs are represented by cylindrical beams mounted on the ground, while the Hertzian contacts among particles are denoted by inter-particle springs. We assume that the wall-mounted springs (i.e., cylindrical beams) impose a potential to our system only in the longitudinal direction of the chain. If the stiffness of the $j$-th wall-mounted spring is $\gamma_j$, and $\gamma_j$ is allowed to vary for $2\leq j\leq n$, then the granular chain system now has $(2n-1)$ degrees of freedom, and it matches the dimension of a $n\times n$ symmetric tridiagonal matrix. Note again that we only vary particles' radii and wall-mounted stiffness values, while keeping 
the elastic modulus and density of beads (i.e., the material) fixed.

\begin{figure}
\includegraphics[width=12cm]{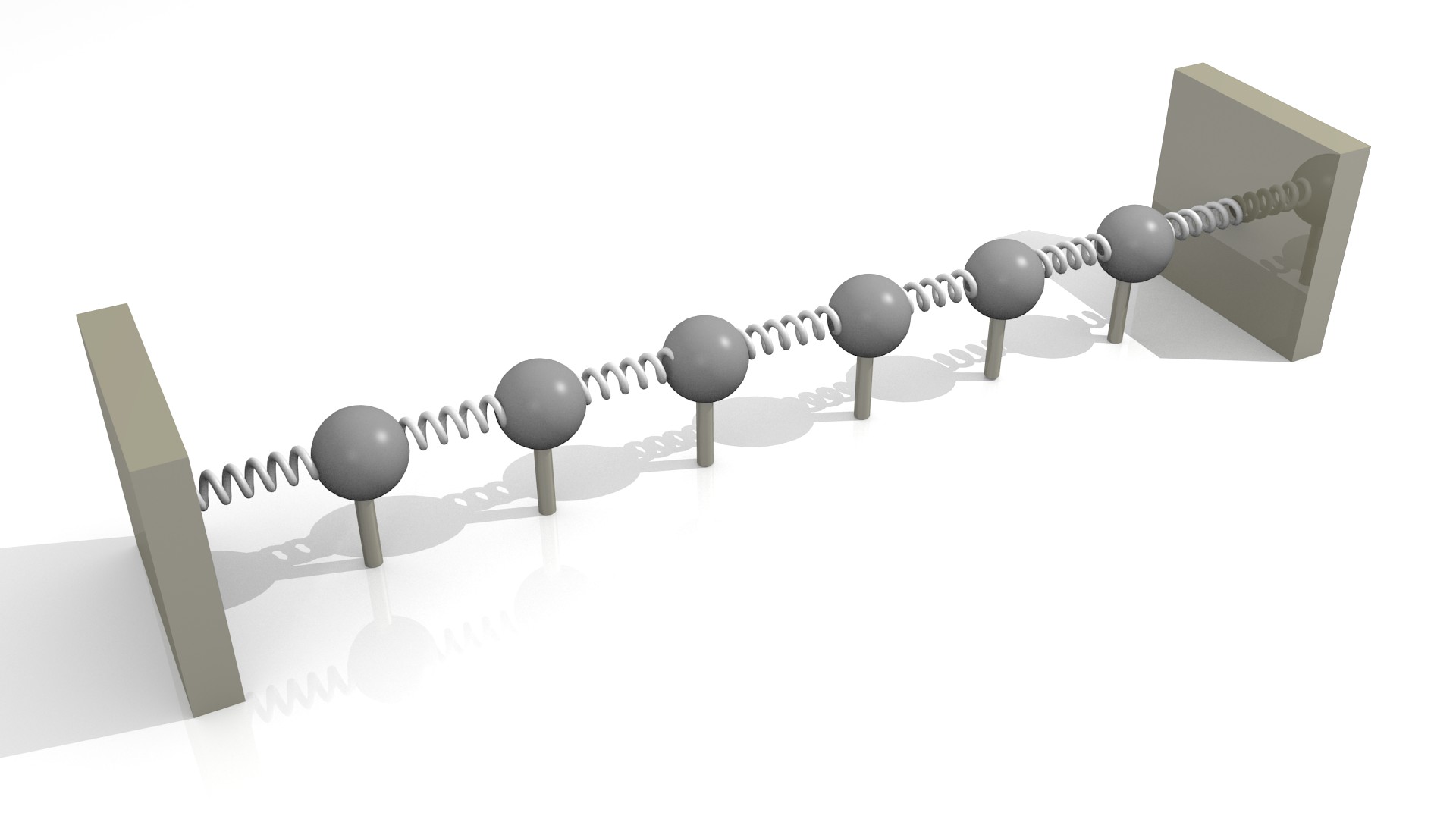}
\caption{General design of the proposed isospectral granular chain. Springs represent the nonlinear Hertz contact interaction between granular beads. Grounded cantilever beam attached to each bead provides a linear stiffness for the longitudinal direction of bead vibration.}
\label{fig1}
\end{figure}

Then, the new equations of motion become:
\begin{eqnarray}
\label{eqn1_gamma}
m_j \ddot{z}_j=k_j(d_j+z_{j-1}-z_j)_+^p-k_{j+1}(d_{j+1}+z_j-z_{j+1})_+^p -\gamma_j z_j, \quad 1\leq j \leq n
\end{eqnarray}
and the linearized system around the steady-state solution is:
\begin{equation}
\label{eqn1_ln_matrix}
M \ddot{\boldsymbol{z}}+B_{\gamma} \boldsymbol{z}=0
\end{equation}
where the mass matrix $M$ remains the same while the stiffness matrix $B_{\gamma}$ has changed its diagonal entries to $K_j+K_{j+1}+\gamma_j$ instead of $K_j+K_{j+1}$ for $1\leq j\leq n$ in $B$.

Although the system in (\ref{eqn1_ln_matrix}) is no longer identical to that in (\ref{eqn0_ln_matrix}) (i.e., now the beads are wall-mounted), its associated matrix $H_{\gamma}=M^{-1/2}B_{\gamma} M^{-1/2}=G^{-1}B_{\gamma} G^{-1}$ will again satisfy (A1)--(A3), provided $\gamma_j$'s are all nonnegative. Then for a matrix $\tilde{H}_{\gamma}$ in $T_{A,n}$ (or $T_{\Lambda}$), we seek a granular chain such that its linearization near the zero solution has $\tilde{H}_{\gamma}$ as its associated matrix. This is possible because the radii and springs in the granular chain give $(2n-1)$ unknowns (note one of the springs, say the first one, is fixed), and the main diagonal and subdiagonals of $\tilde{H}_{\gamma}$ also give $(2n-1)$ equations.

Suppose a granular chain system (NS-2) has radii $(\tilde{r}_1,\tilde{r}_2,...,\tilde{r}_n)$, springs $(\gamma_1,\tilde{\gamma}_2,\tilde{\gamma}_3,...,\tilde{\gamma}_n)$ and precompressions $(\tilde{d}_1,\tilde{d}_2,...,\tilde{d}_{n+1})$. If (NS-2) has a zero solution, then $\tilde{k}_1 \tilde{d}_1^p=\tilde{k}_2 \tilde{d}_2^p=...=\tilde{k}_{n+1} \tilde{d}_{n+1}^p=C$ where $\tilde{k}_j$ represents the stiffness of elastic force between $(j-1)$-th and $j$-th beads. This is under the assumption that the mounted springs are unstretched at the initial compressive state of the chain. In the linearized system of (NS-2) near the zero solution, we have:
\begin{eqnarray}
\tilde{m}_j &=& \frac{4}{3}\pi\rho \tilde{r}_j^3, \quad 1\leq j \leq n \\
\tilde{K}_1 &=& W (\tilde{r}_1)^{1/3}, \\
\tilde{K}_j &=& W (\frac{\tilde{r}_{j-1}\tilde{r}_j}{\tilde{r}_{j-1}+\tilde{r}_j})^{1/3}, \quad 2\leq j\leq n\\
\tilde{K}_{n+1} &=& W (\tilde{r}_n)^{1/3}
\end{eqnarray}
where $W=p C^{1-1/p} (\frac{4}{3}E_*)^{1/p}=\frac{3}{2} C^{\frac{1}{3}} (\frac{4}{3}E_*)^{\frac{2}{3}}$.

Substituting these expressions back into $\tilde{H}_{\gamma}=\tilde{M}^{-1/2}\tilde{B}_{\gamma} \tilde{M}^{-1/2}$ where $\tilde{H}_{\gamma}$ is given, we obtain the equations as follows for $\tilde{\boldsymbol{r}}=(\tilde{r}_1,\tilde{r}_2,...,\tilde{r}_n)$ and $\tilde{\boldsymbol{\gamma}}=(\tilde{\gamma}_2,\tilde{\gamma}_3,...,\tilde{\gamma}_n)$:
\begin{eqnarray}
\label{eqn_fg_1}
f_1 &:=& \tilde{H}_{\gamma}(1,1)- \frac{W( (\tilde{r}_1)^{1/3}+ (\frac{\tilde{r}_{1}\tilde{r}_2}{\tilde{r}_{1}+\tilde{r}_2})^{1/3} ) + \gamma_1 }{\frac{4}{3}\pi\rho \tilde{r}_1^3} \\
\label{eqn_fg_2}
f_j &:=& \tilde{H}_{\gamma}(j,j)- \frac{W( (\frac{\tilde{r}_{j-1}\tilde{r}_{j}}{\tilde{r}_{j-1}+\tilde{r}_{j}})^{1/3}+ (\frac{\tilde{r}_{j}\tilde{r}_{j+1}}{\tilde{r}_{j}+\tilde{r}_{j+1}})^{1/3} ) + \tilde{\gamma}_j}{\frac{4}{3}\pi\rho \tilde{r}_j^3}, \quad 2\leq j \leq n-1 \\
\label{eqn_fg_3}
f_n &:=& \tilde{H}_{\gamma}(n,n)- \frac{W( (\tilde{r}_n)^{1/3}+ (\frac{\tilde{r}_{n-1}\tilde{r}_n}{\tilde{r}_{n-1}+\tilde{r}_n})^{1/3} ) + \tilde{\gamma}_n }{\frac{4}{3}\pi\rho \tilde{r}_n^3} \\
\label{eqn_fg_4}
g_j &:=& \tilde{H}_{\gamma}(j,j+1)+ \frac{W (\frac{\tilde{r}_{j}\tilde{r}_{j+1}}{\tilde{r}_{j}+\tilde{r}_{j+1}})^{1/3} }{\frac{4}{3}\pi\rho (\tilde{r}_j)^{3/2}(\tilde{r}_{j+1})^{3/2} }, \quad 1\leq j\leq n-1
\end{eqnarray}

\begin{lemma}
\label{lemma_det}
If $(\tilde{r}_1,\tilde{r}_2,...,\tilde{r}_n)$ are positive and $\gamma_1\geq 0$ in equations (\ref{eqn_fg_1} - \ref{eqn_fg_4}), then $|\frac{\partial (\boldsymbol{f}, \boldsymbol{g})}{\partial (\tilde{\boldsymbol{r}}, \tilde{\boldsymbol{\gamma}})}|\neq 0$. (See Appendix for proof)
\end{lemma}

\begin{lemma}
\label{lemma_granular}
Suppose $\tilde{H}_{\gamma}\in T_{\Lambda}$ corresponds to $\tilde{\boldsymbol{x}}$ in $S^{n-1}_{>0}$ and $\gamma_1\geq 0$ in $(\boldsymbol{f}(\tilde{\boldsymbol{r}},\tilde{\boldsymbol{\gamma}},\tilde{H}_{\gamma}), \boldsymbol{g}(\tilde{\boldsymbol{r}},\tilde{\boldsymbol{\gamma}},\tilde{H}_{\gamma}))=0$, i.e., equations (\ref{eqn_fg_1} - \ref{eqn_fg_4}). If  $(\boldsymbol{f}({\boldsymbol{r}},{\boldsymbol{\gamma}},{H}_{\gamma}), \boldsymbol{g}({\boldsymbol{r}},{\boldsymbol{\gamma}},{H}_{\gamma}))=0$ where ${\boldsymbol{r}}$ is a positive vector and ${H}_{\gamma}$ corresponds to ${\boldsymbol{x}}$ in $S^{n-1}_{>0}$, then for some open set $O$ in $S^{n-1}_{>0}$ containing ${\boldsymbol{x}}$, there exists a unique $(\tilde{\boldsymbol{r}}(\tilde{\boldsymbol{x}}),\tilde{\boldsymbol{\gamma}}(\tilde{\boldsymbol{x}}))$ satisfying $(\boldsymbol{f}(\tilde{\boldsymbol{r}}(\tilde{\boldsymbol{x}}),\tilde{\boldsymbol{\gamma}}(\tilde{\boldsymbol{x}}),\tilde{H}_{\gamma}(\tilde{\boldsymbol{x}})), \boldsymbol{g}(\tilde{\boldsymbol{r}}(\tilde{\boldsymbol{x}}),\tilde{\boldsymbol{\gamma}}(\tilde{\boldsymbol{x}}),\tilde{H}_{\gamma}(\tilde{\boldsymbol{x}})))=0$ for $\tilde{\boldsymbol{x}}\in O$.
\end{lemma}

\begin{proof}
By Lemma \ref{lemma_continuous}, we write $(\boldsymbol{f}(\tilde{\boldsymbol{r}},\tilde{\boldsymbol{\gamma}},\tilde{H}_{\gamma}), \boldsymbol{g}(\tilde{\boldsymbol{r}},\tilde{\boldsymbol{\gamma}},\tilde{H}_{\gamma}))=0$ as $(\boldsymbol{f}(\tilde{\boldsymbol{r}},\tilde{\boldsymbol{\gamma}},\tilde{\boldsymbol{x}}), \boldsymbol{g}(\tilde{\boldsymbol{r}},\tilde{\boldsymbol{\gamma}},\tilde{\boldsymbol{x}}))=0$. Then, we apply Lemma \ref{lemma_det} and the Implicit Function Theorem.
\end{proof}

If a granular chain (NS-1) with $(r_1,r_2,...,r_n)$ and $(\gamma_1,\gamma_2,...,\gamma_n)$ gives matrix $H_{\gamma}$ in $T_{\Lambda}$ and $\phi(H_{\gamma})=\boldsymbol{x}$, then by Lemma \ref{lemma_granular}, there exists a family of granular chains with the same eigenfrequencies in the linear limit as (NS-1). We note that $\tilde{\boldsymbol{r}}$ obtained here will always be positive, however $\boldsymbol{\tilde{\gamma}}$ may contain negative components, which will violate the positivity
of spring constants. That is to say, though we are in general free to vary $\boldsymbol{\tilde{x}}\in S_{>0}^{n-1}$ to construct isospectral granular chains, perhaps only a part of the values for $\boldsymbol{\tilde{x}}$ can lead to systems with positive $\boldsymbol{\tilde{\gamma}}$ such that they are physically realizable. In this work, we have chosen our
examples carefully so that they all have positive $\tilde{\boldsymbol{\gamma}}$. However, it is certainly interesting to
identify the precise conditions for positive $\tilde{\boldsymbol{\gamma}}$.
This is an important open problem that will be  
deferred to future studies.


\subsubsection*{Example 1: isospectral granular chains}

In Fig.~\ref{fig2}, we show an example of finding granular chains whose linearized models have the same eigenfrequencies. Suppose we start with a generalized granular chain system with $\rho=1$, $E_*=1$ and
\begin{itemize}
\item radii $r_j=1$ for all $1\leq j\leq 10$,
\item springs $\gamma_1=0$ and $\gamma_j=1$ for all $2\leq j\leq 10$,
\item precompressions satisfying $k_1 d_1^p=k_2 d_2^p=...=k_{11} d_{11}^p=C$ where $C=1$,
\end{itemize}
then by Lemma~\ref{lemma_continuous} its corresponding matrix $H$ has the image point
\[
\boldsymbol{x}=\phi(H)=(0.08415, 0.1667, 0.2451, 0.3158, 0.3743, 0.4148, 0.4291, 0.4057, 0.3295, 0.1895) \in S^{9}_{>0}.
\]
Then we move $\boldsymbol{x}$ along the homotopic path (ensuring that
we stay on the unit sphere)
\[
\Gamma_1(t)=\frac{\boldsymbol{x}\cdot(1-t)+\boldsymbol{y}\cdot t}{\| \boldsymbol{x}\cdot(1-t)+\boldsymbol{y}\cdot t \|_{l_2}}
\]
where
\[
\quad y_j=\frac{\sin(\frac{j\pi}{11})}{\sqrt{\sum_{j=1}^{10}\sin^2(\frac{j\pi}{11})}}, \quad 1\leq j \leq n
\]
to obtain a sequence of points $\{\boldsymbol{x}^{(l)}\}$ with $\boldsymbol{x}^{(0)}=\boldsymbol{x}$ and compute the corresponding tridiagonal matrices $\{H^{(l)}\}$ in $T_{\Lambda}$. By construction, all these tridiagonal matrices are isospectral. Assuming $\rho$, $\gamma_1$, and $C$ remain the same for all systems, we use Newton's method to solve~(\ref{eqn_fg_1})--(\ref{eqn_fg_4}), aiming to find the granular chain with $\boldsymbol{r}^{(l)}$ and $\boldsymbol{\gamma}^{(l)}$ for given $H^{(l)}=\phi^{-1}(\boldsymbol{x}^{(l)})$. If the steps in $\{\boldsymbol{x}^{(l)}\}$ are chosen small enough, Lemmas~\ref{lemma_det} and~\ref{lemma_granular} tell us that $\boldsymbol{r}^{(l-1)}$ and $\boldsymbol{\gamma}^{(l-1)}$ can serve as an initial guess for the $l$-th iteration in Newton's method. In particular,
the corresponding granular chain system we obtained for $\boldsymbol{x}^{(20)}=\boldsymbol{y}$ has:
\begin{itemize}
\item radii $\boldsymbol{r}^{(20)}=(0.9421, 1.0523, 0.9493, 1.0530, 0.9493, 1.0530, 0.9493, 1.0530, 0.9491, 1.0193)$
\item springs $\boldsymbol{\gamma}^{(20)}=(0, 1.6389, 0.4390, 1.6523, 0.4404, 1.6526, 0.4404, 1.6523, 0.4440, 0.6306)$.
\end{itemize}

In Fig.~\ref{fig3}, we arrive at the same image point $\boldsymbol{x}^{(20)}$ via a different path
\[
\Gamma_2(t)=\frac{\boldsymbol{x}+(\boldsymbol{y}-\boldsymbol{x})(3 t-2t^2)}{\| \boldsymbol{x}+(\boldsymbol{y}-\boldsymbol{x})(3 t-2t^2) \|_{l_2}}
\]
in $S^{9}_{>0}$, finding that the corresponding granular chain system at the destination is the same as what we obtained using the first path.

\begin{figure}[t]
\centering
\begin{tabular}{cc}
\includegraphics[width=6cm]{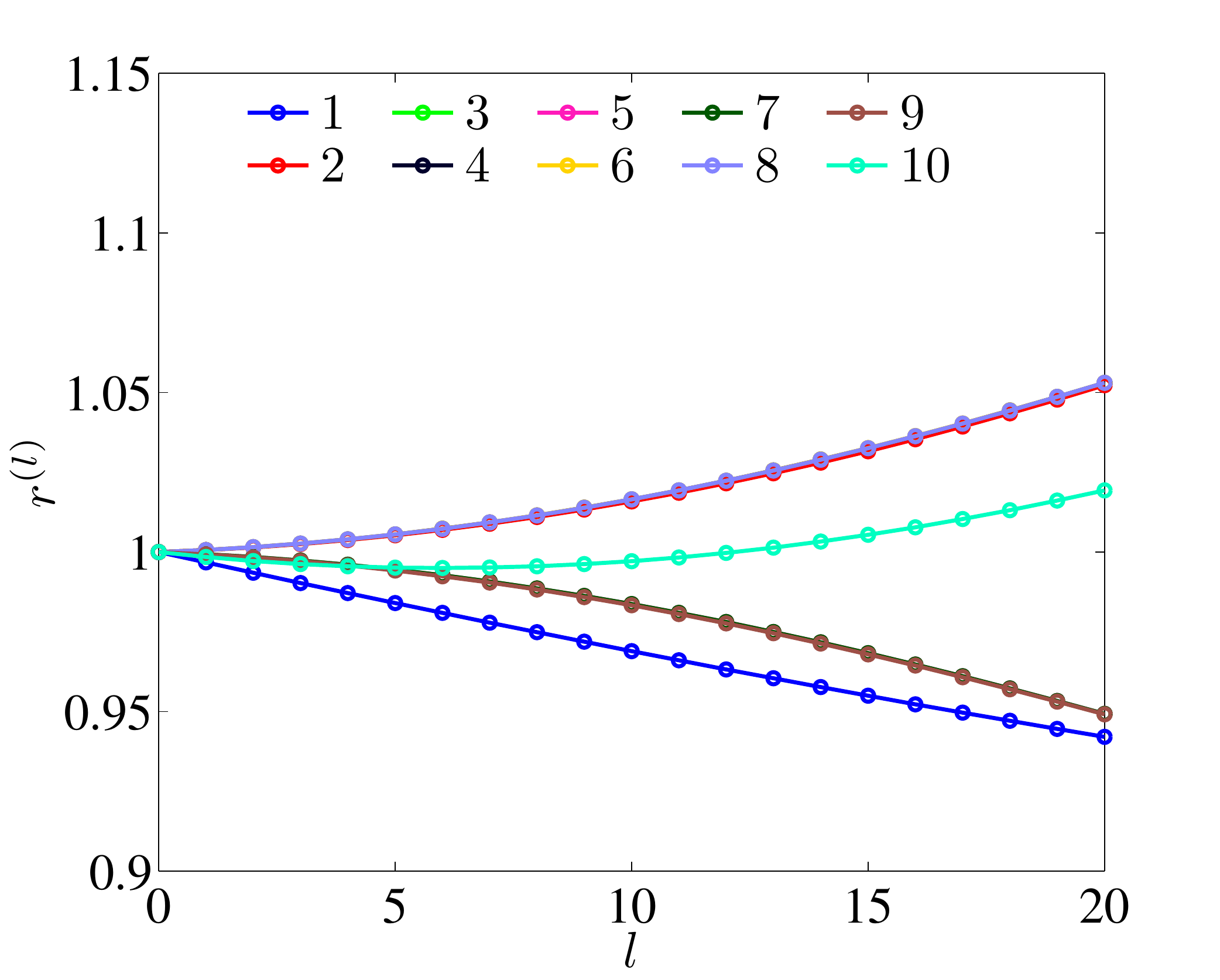}
\includegraphics[width=6cm]{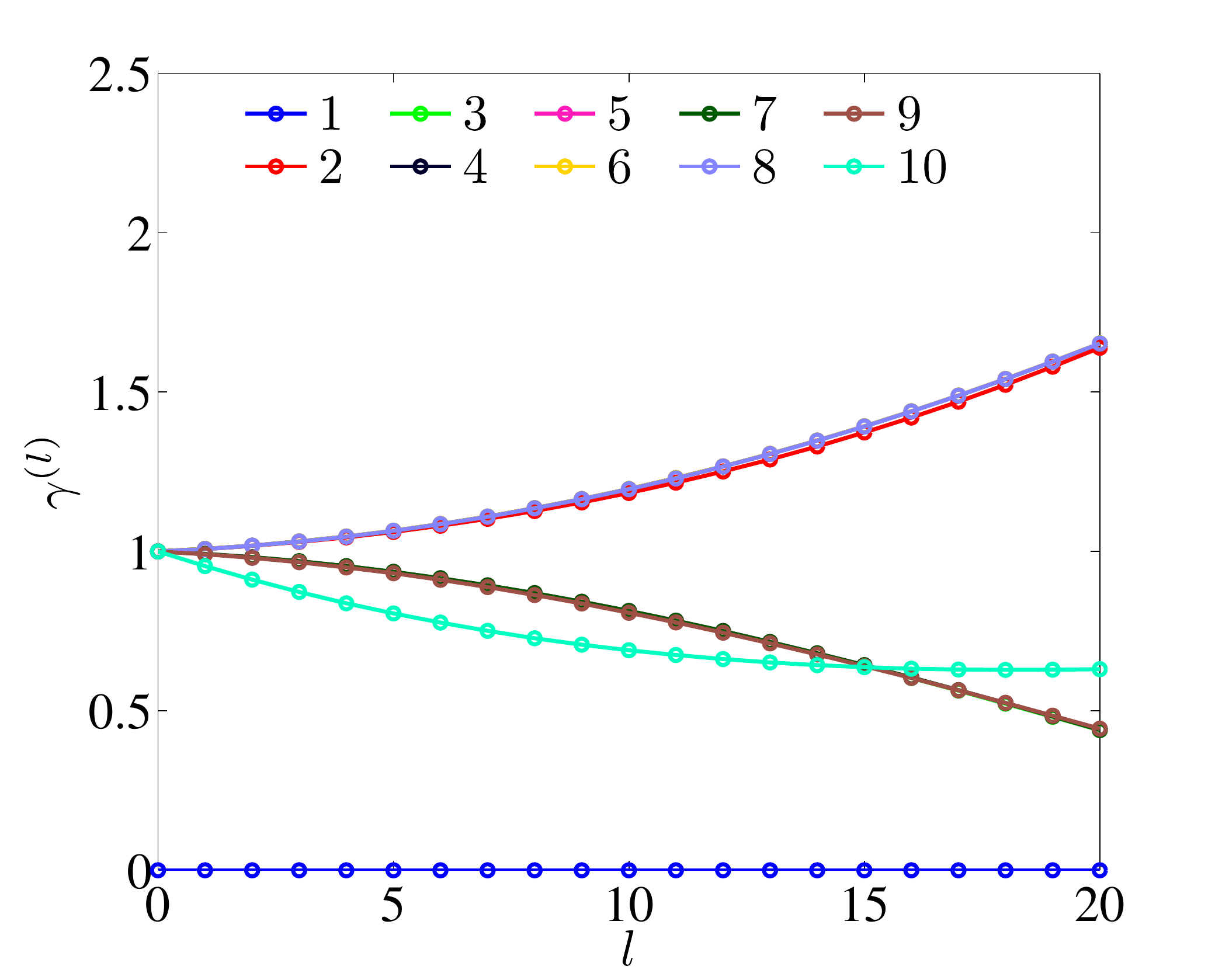}\\
\includegraphics[width=6cm]{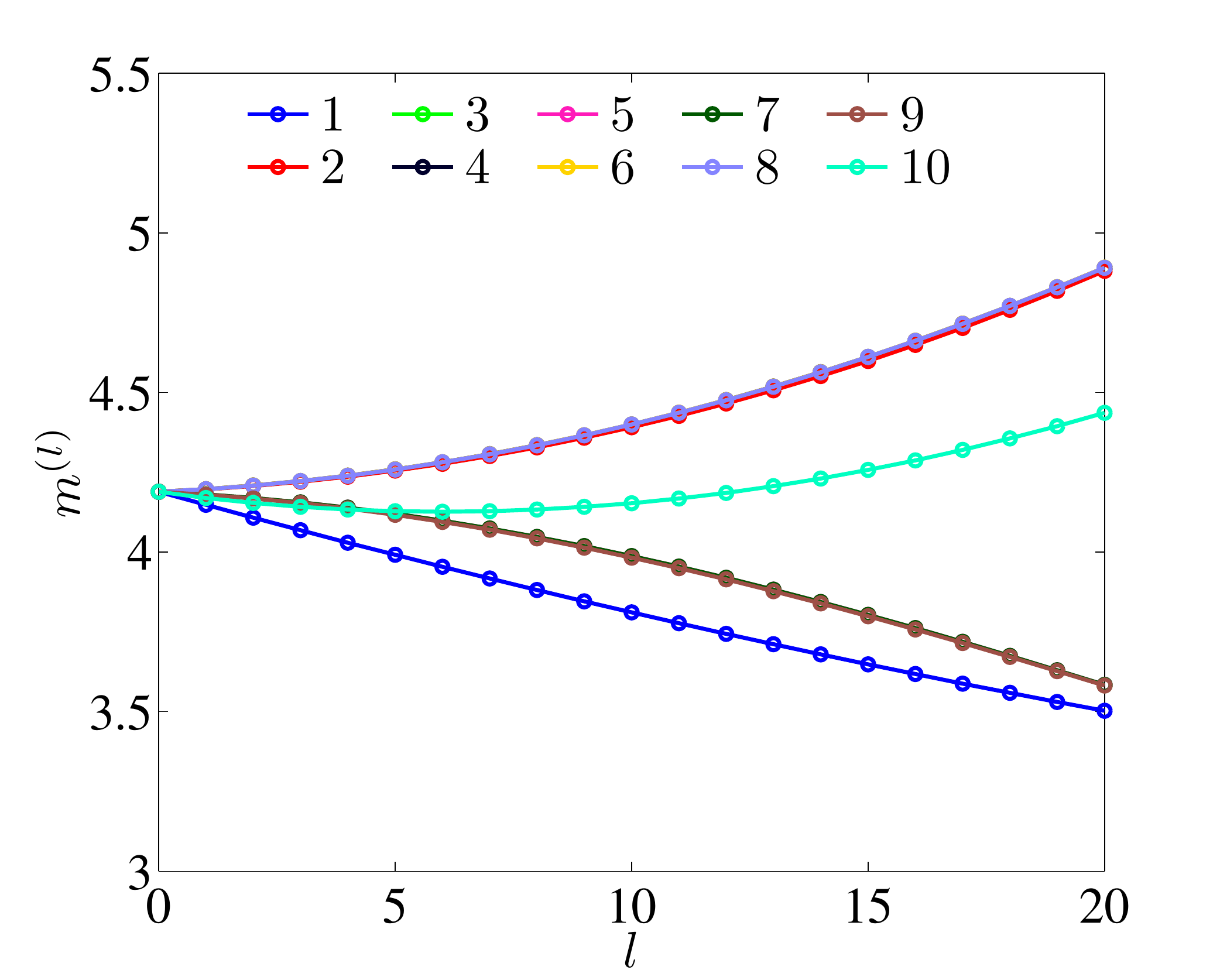}
\includegraphics[width=6cm]{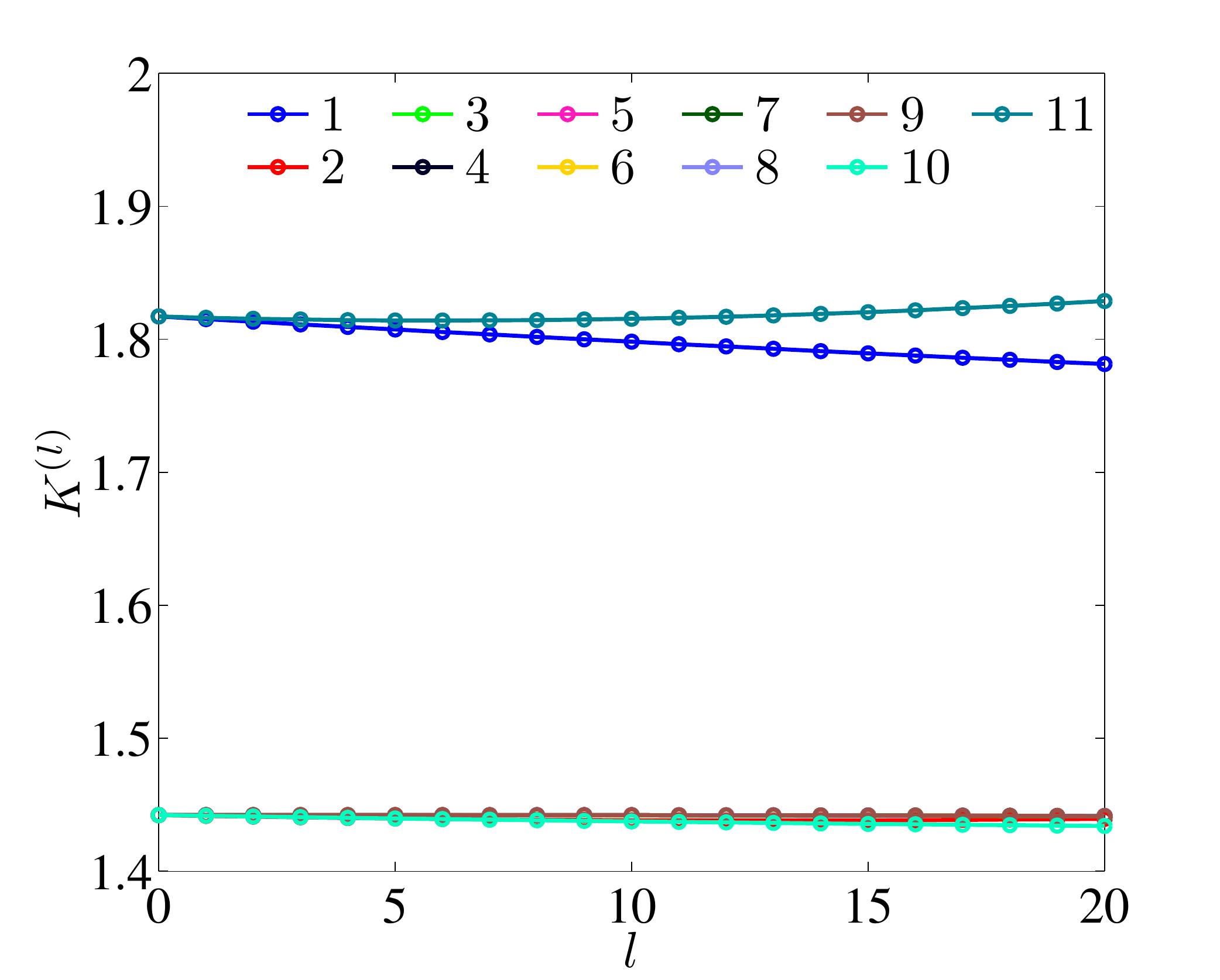}\\
\includegraphics[width=6cm]{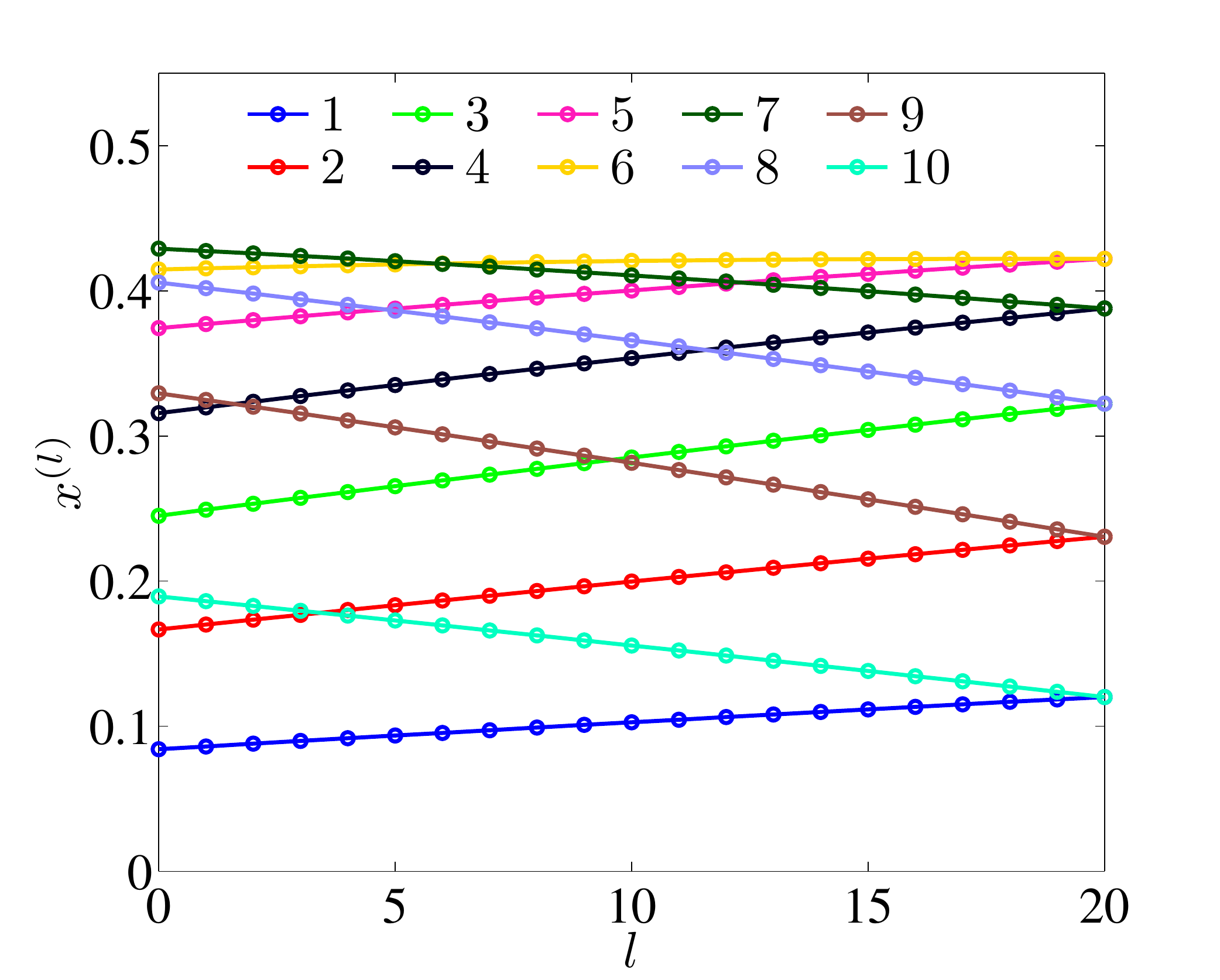}
\end{tabular}
\caption{Top row: radii (left panel) and springs (right panel) in granular chain systems; middle row: masses (left panel) and spring coefficients (right panel) in the linearized systems; bottom: corresponding image points in $S^{9}_{>0}$. In each of the panels, each line stands for a variable in a sequence of isospectral systems. The circles in each column correspond to different variables in the same system. As $\boldsymbol{x}$ goes to $\boldsymbol{y}$ in this example, we particularly observe that $r_{3}^{(l)}$, $r_{5}^{(l)}$, $r_{7}^{(l)}$ and $r_{9}^{(l)}$ are almost identical while $r_{2}^{(l)}$, $r_{4}^{(l)}$, $r_{6}^{(l)}$, $r_{8}^{(l)}$ are quantitatively identical. Similarly, $\gamma_{3}^{(l)}$, $\gamma_{5}^{(l)}$, $\gamma_{7}^{(l)}$ and $\gamma_{9}^{(l)}$ are close while the lines of $\gamma_{2}^{(l)}$, $\gamma_{4}^{(l)}$, $\gamma_{6}^{(l)}$ and $\gamma_{8}^{(l)}$ almost collide. $\boldsymbol{m^{(l)}}$ and $\boldsymbol{K^{(l)}}$ bear similar features since they are functions of $\boldsymbol{r^{(l)}}$. 
}
\label{fig2}
\end{figure}

\begin{figure}[t]
\centering
\begin{tabular}{cc}
\includegraphics[width=6cm]{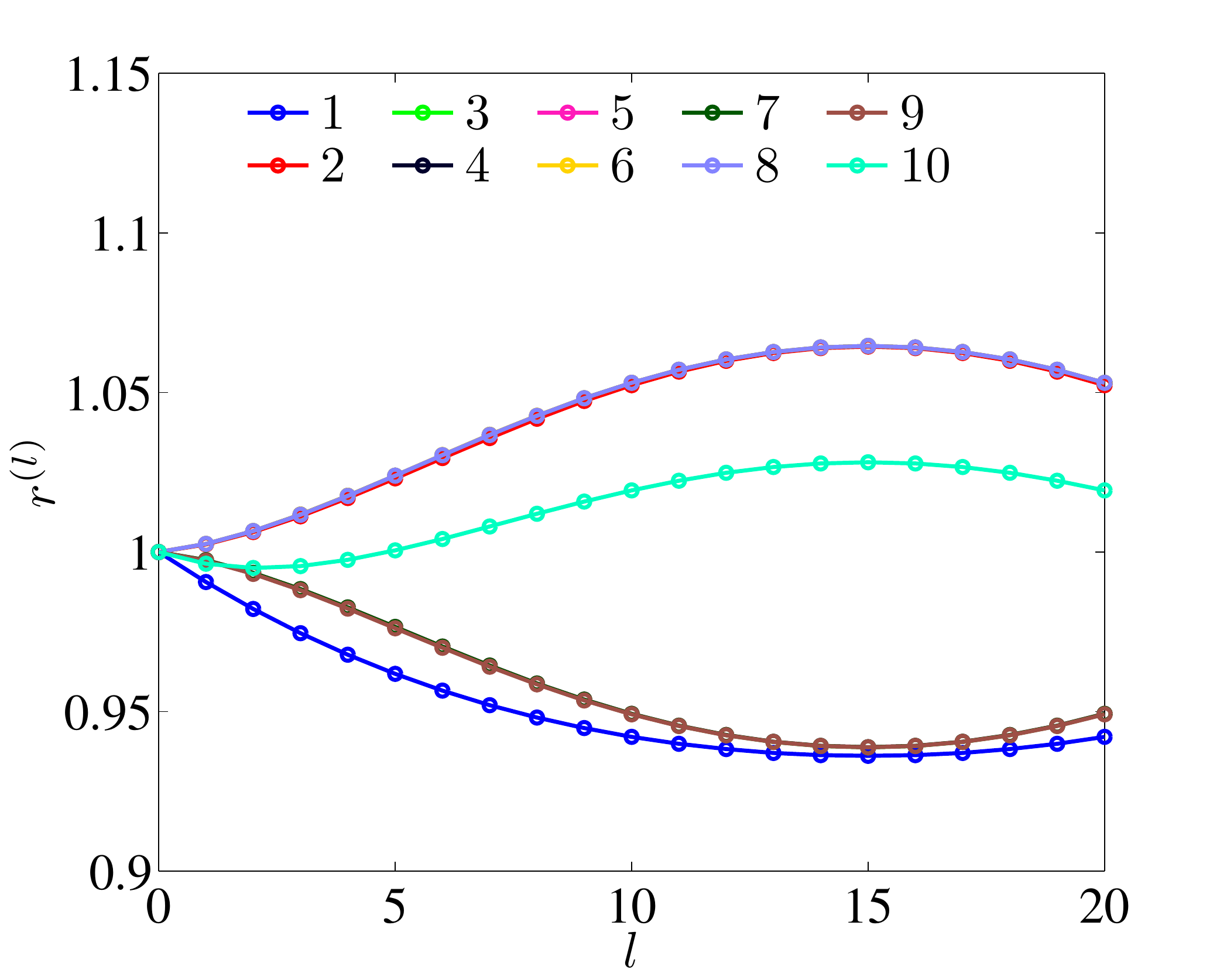}
\includegraphics[width=6cm]{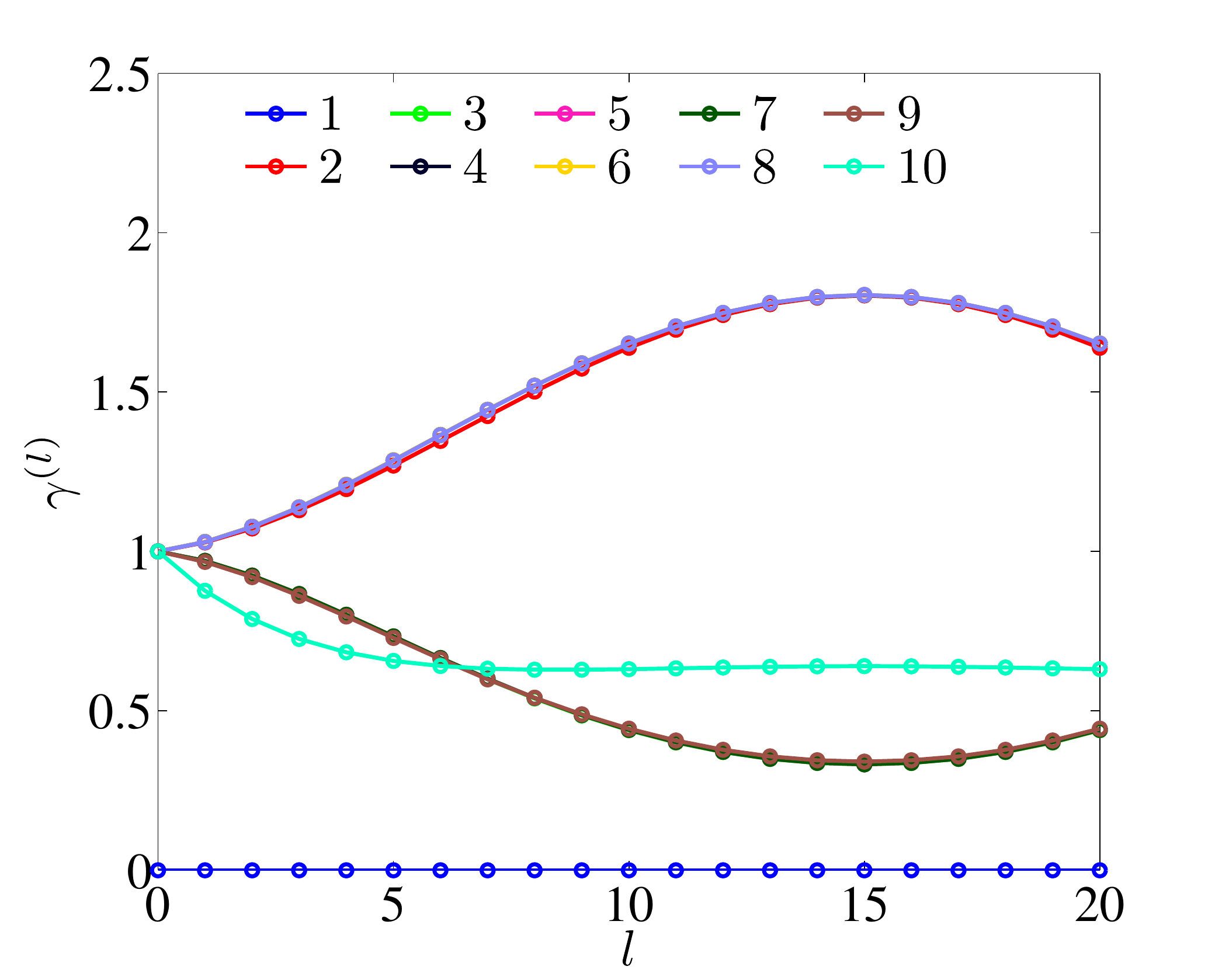}\\
\includegraphics[width=6cm]{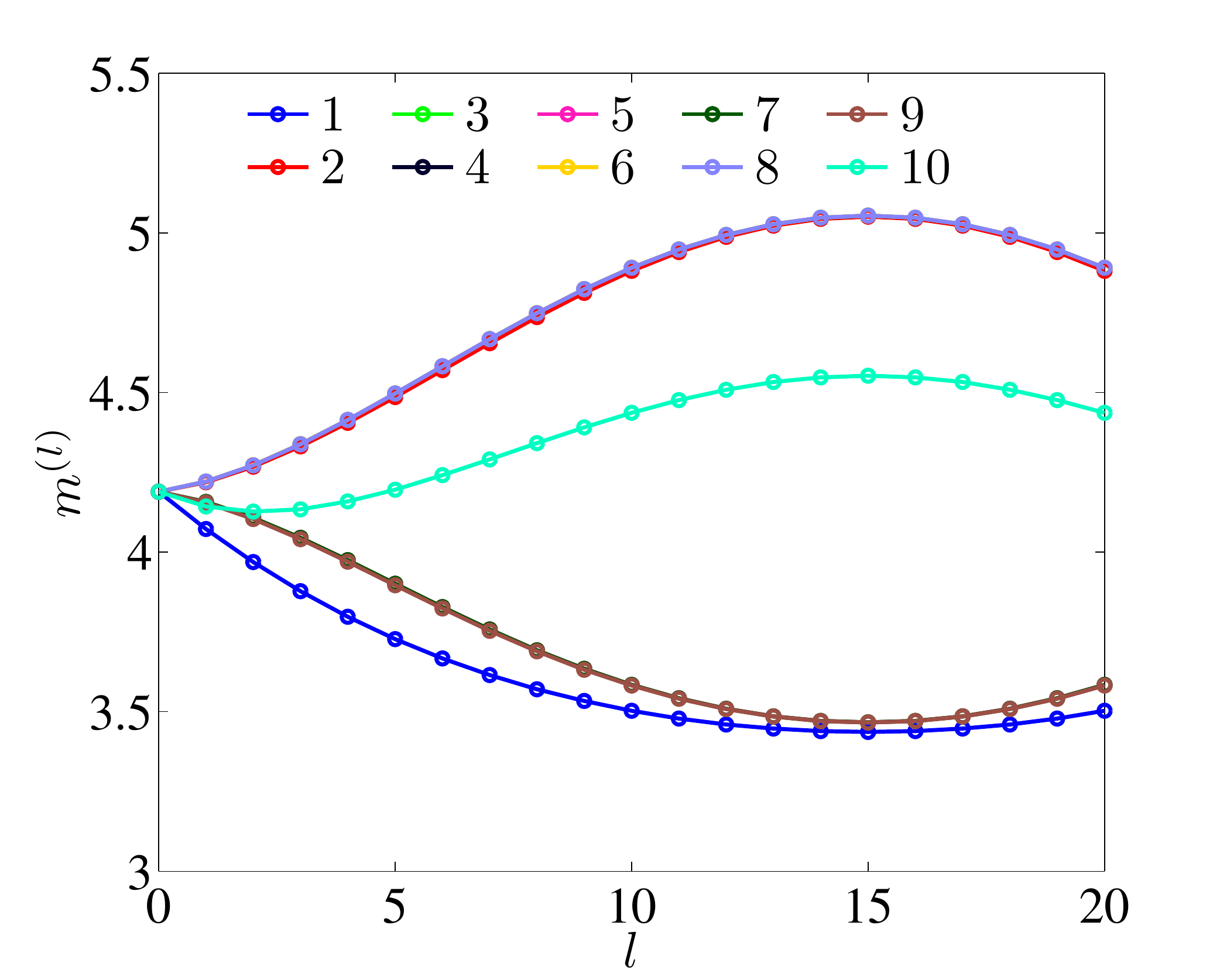}
\includegraphics[width=6cm]{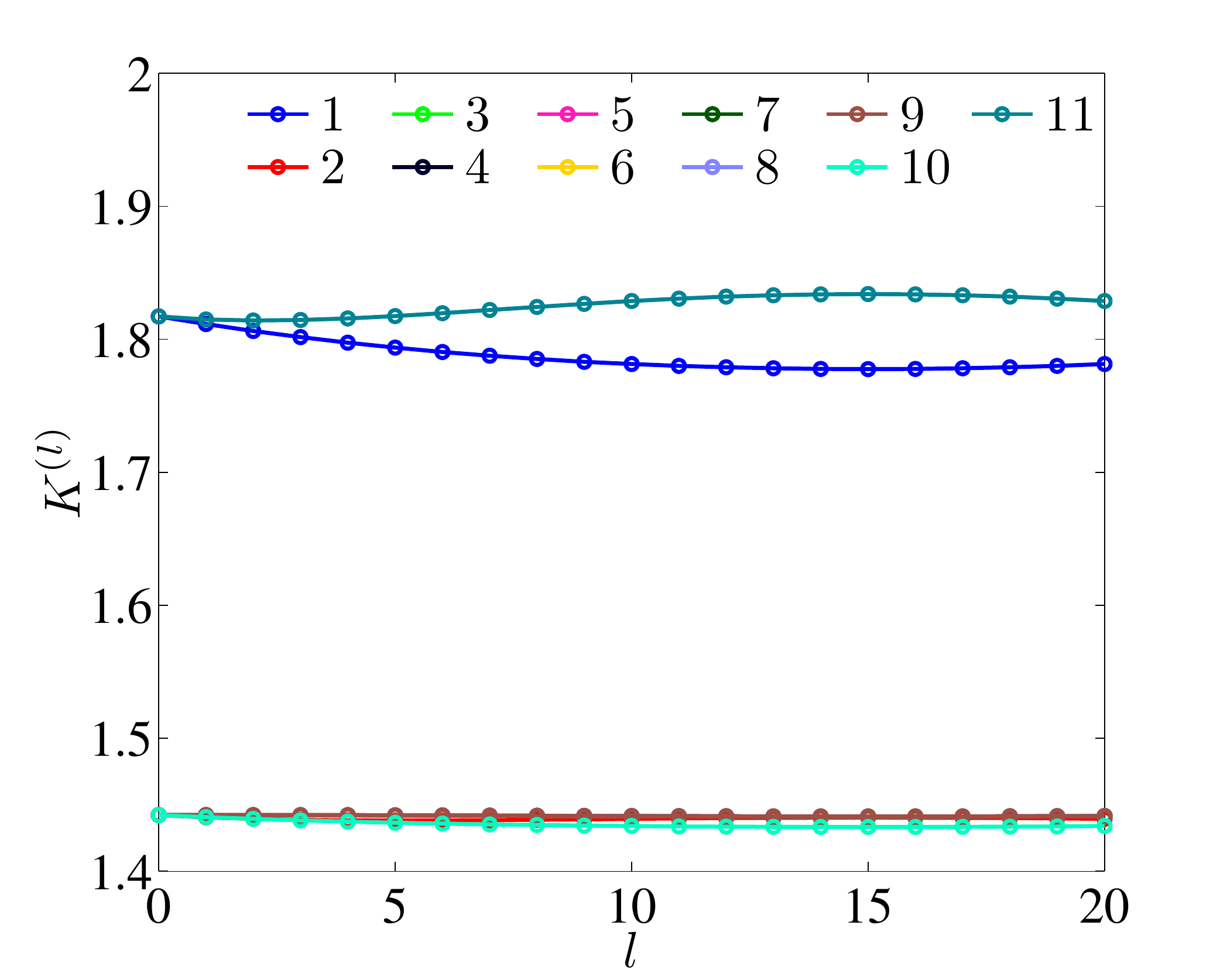}\\
\includegraphics[width=6cm]{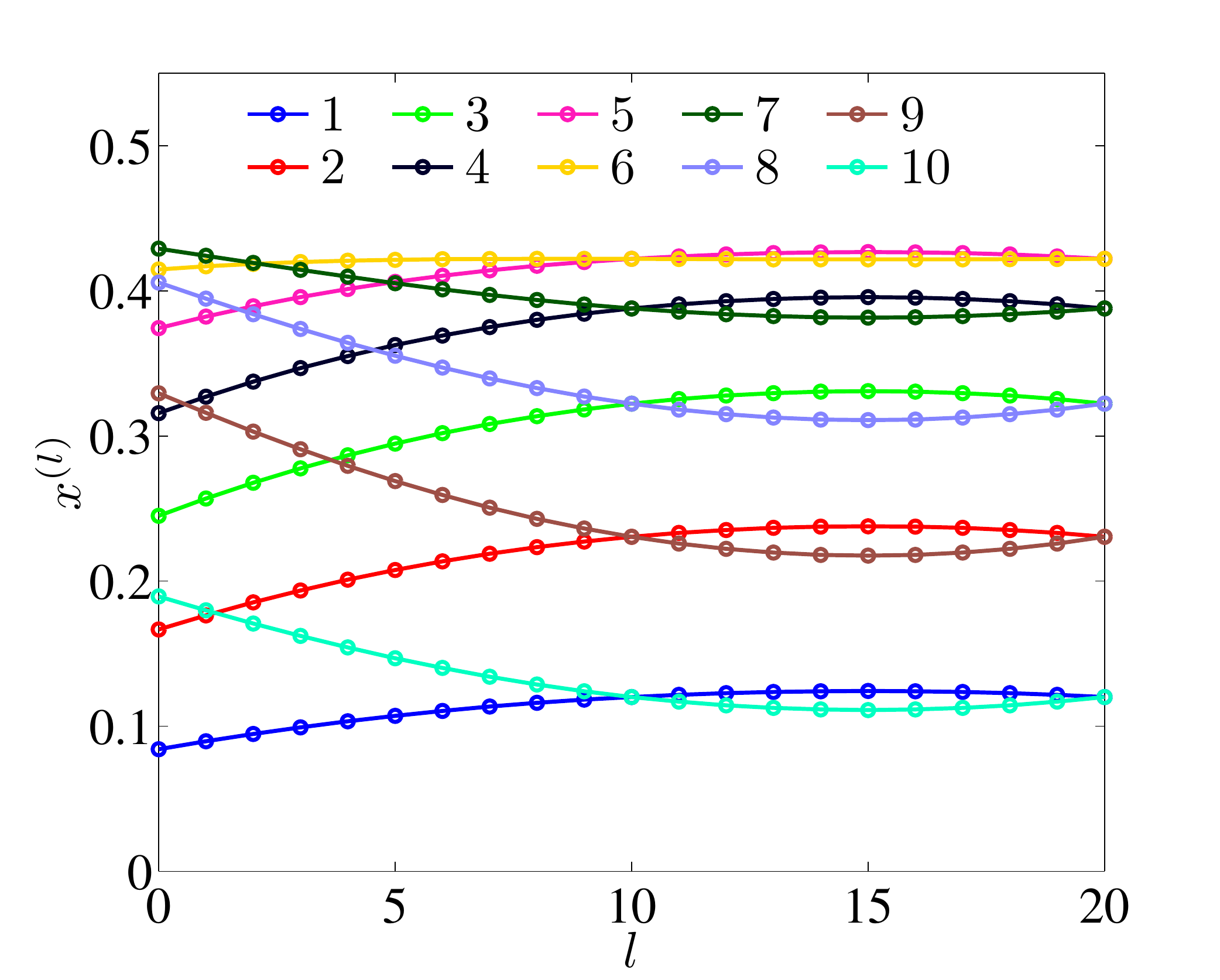}
\end{tabular}
\caption{Top row: radii (left panel) and springs (right panel) in granular chain systems; middle row: masses (left panel) and spring coefficients (right panel) in the linearized systems; bottom: corresponding image points in the space $S^{9}_{>0}$. The meaning of the lines and circles follows that for Fig.~\ref{fig2}. It is shown that we arrive at the same granular chain as in Fig.~\ref{fig2} at $\boldsymbol{y}$, though via a different path.}
\label{fig3}
\end{figure}

\subsection{Frequency addition and removal}
In the previous section, we talked about the process of constructing a family of granular chain systems from an initial one, while keeping identical all eigenfrequencies of its linearized model. In this section, we will discuss the possibility of constructing granular chain systems if we add or remove eigenvalues from the spectrum of its linearized model.

Suppose now we have a granular chain system with radii $\{r_j\}_{1\leq j\leq n}$ and local stiffnesses $\{\gamma_j\}_{2\leq j\leq n}$.
After linearizing the system around the zero solution, we obtain a tridiagonal matrix $H$ whose spectrum is denoted by $\Lambda$. In order to directly remove an eigenvalue $\lambda_k$ of $H$ from $\Lambda$, we would like to make $H(n-1, n)=H(n,n-1)=0$ and $H(n,n)=\lambda_k$ where $\lambda_k\in \Lambda$.
This enables us to decouple one eigenvalue from the system, in the sense that once the $n$-th column and $n$-th row are truncated from $H$ matrix, we can 
remove that eigenvalue without affecting the remaining ones.
Since the matrix $H$ represents the linearization of a granular chain consisting of spherical beads and local potentials, then $H(n-1,n)\to 0$ implies $r_n\to \infty$ or $r_{n-1}\to\infty$. Since $r_{n-1}\to\infty$ also leads to $H(n-2,n-1)=H(n-1,n-2)\to 0$, this choice will not fit the situation of only separating the eigenvalue $\lambda_k$. On the other hand, if we set $r_n\to \infty$ and $\gamma_n \to \lambda_k \frac{4}{3}\pi \rho r_n^3$, then $\frac{ r_n r_{n-1} }{ r_n + r_{n-1} } \to r_{n-1}$ and $H(n,n)\to \lambda_k$. As a result, we will have $H \to \left(
    \begin{array}{cc}
        {H}_{n-1} & 0\\
        0 & \lambda_k
    \end{array}
\right)$ where $H_{n-1}$ is a $(n-1)\times(n-1)$ matrix in $T_{\Lambda \backslash \lambda_k }$.

\begin{proposition}
\label{proposition_3}
If there is a sequence of matrices $\{H^{(l)}\}_{l\in \mathbb{N}}$ converging to $\left(
    \begin{array}{cc}
        {H}_{n-1} & 0\\
        0 & \lambda_k
    \end{array}
\right)$ in $T_{\Lambda}$ and each of the matrix $H^{(l)}$ corresponds to a granular chain model with $\{r_j^{(l)}\}_{1\leq j\leq n}$ and $\{\gamma_j^{(l)}\}_{2\leq j\leq n}$, then $\lim_{l\to \infty} r_n^{(l)} = \infty $ and $\lim_{l\to \infty} \frac{\gamma_n^{(l)}}{(r_n^{(l)})^3} = \lambda_k \frac{4}{3}\pi \rho$.
\end{proposition}
~\\
In fact, this result makes intuitive sense via the interpretation
that the $n$-th bead will gradually become harder and harder to move as $r_n\to \infty$ and $\gamma_n\to\infty$ and effectively become a ``wall'' in the end. This makes the $n$-th bead decoupled from the chain with its own eigenfrequency $\lambda$ while the other $(n-1)$ beads will actually represent a shorter granular chain that corresponds to the matrix ${H}_{n-1}$. 
In order to find a sequence of matrices described in Proposition~\ref{proposition_3}, we consider the behavior of their image points in $S_{>0}^{n-1}$ and state the following proposition, which we have numerically examined though a proof is not given here.

\begin{proposition}
\label{proposition_4}
Suppose a sequence of matrices $\{H^{(l)}\}_{l\in \mathbb{N}}$ is in $T_{\Lambda}$ where $\Lambda=\{ \lambda_1 > \lambda_2 > ... > \lambda_n >0 \}$ and their image points in $S_{>0}^{n-1}$ is $\boldsymbol{x}^{(l)}=(x_1^{(l)}, x_2^{(l)}, ... , x_n^{(l)})$. If $\lim_{l\to\infty} x_k^{(l)} = 0$ and $\lim_{l\to\infty}x_j^{(l)} = y_j >0$ for $j\neq k$, then $\lim_{l\to\infty}H^{(l)} = \left(
    \begin{array}{cc}
        {H}_{n-1} & 0\\
        0 & \lambda_k
    \end{array}
\right)$ where $H_{n-1} \in T_{\Lambda \backslash \lambda_k }$.
\end{proposition}

~\\
With the results stated above, we introduce the following algorithm to find a family of granular chain systems with the eigenvalues of their linearized models gradually removed:
\begin{enumerate}
\item Start from a granular chain with $\{r_j\}_{1\leq j\leq n}$ and $\{\gamma_j\}_{2\leq j\leq n}$, whose linearization around the zero solution corresponds to matrix $H$. Suppose the spectrum $H$ is $\Lambda=\{ \lambda_1 > \lambda_2 > ... > \lambda_n >0 \}$ and its image point in $S_{>0}^{n-1}$ is $\boldsymbol{x}=(x_1,x_2,...,x_n)$.
\item Pick a sequence of points $\boldsymbol{x}^{(l)}=(x_1^{(l)}, x_2^{(l)}, ... , x_n^{(l)})$ in $S_{>0}^{n-1}$ such that $\lim_{l\to\infty} x_k^{(l)} = 0$ and $\lim_{l\to\infty}x_j^{(l)} = y_j >0$ for $j\neq k$. For the numerical computation to run smoothly, we require that $\boldsymbol{x}^{l}$ is chosen close enough to $\boldsymbol{x}^{l-1}$ for each $l$ in the sequence.
\item  For each $\boldsymbol{x}^{(l)}$, by Lemma~\ref{lemma_granular}, we can use numerical methods such as Newton's method to compute the granular chain with $\{r_j^{(l)}\}_{1\leq j\leq n}$ and $\{\gamma_j^{(l)}\}_{2\leq j\leq n}$. To be more specific, we apply Newton's method to the equations (\ref{eqn_fg_1} - \ref{eqn_fg_4}) using $\{r_j^{(l)}\}_{1\leq j\leq n}$ and $\{\gamma_j^{(l)}\}_{2\leq j\leq n}$ as variables. According to Lemma~\ref{lemma_granular}, if we choose $\boldsymbol{x}^{(l)}$ close enough to $\boldsymbol{x}^{(l-1)}$, then by continuity $\{r_{j}^{(l-1)}\}_{1\leq j\leq n}$ and $\{\gamma_{j}^{(l-1)}\}_{2\leq j\leq n}$ should be a good guess for the solution and the Newton's method will have a second-order convergence.

\item When $x_k^{(l)}$ is close enough to zero or $r_n^{(l)}$ and $\gamma_n^{(l)}$ are large enough, we drop the $n$-th bead of the chain so that the length of the chain decreases by $1$.

    To be more specific, we find the matrix $H_{n-1}\in T_{\Lambda\backslash \lambda_k}$ such that it corresponds to $(y_1,y_2,...,y_{k-1},y_{k+1},...,y_{n})$ in $S_{>0}^{n-2}$. Then we solve for a granular chain of length $n-1$ from $H_{n-1}$ while using $\{r_j^{(l)}\}_{1\leq j\leq n-1}$ and $\{\gamma_j^{(l)}\}_{2\leq j\leq n-1}$ as an initial guess.
\end{enumerate}

By repeating the steps 1-4, we can remove eigenvalues from the spectrum $\Lambda$ one by one in any order and decrease the length of the chain accordingly.

Reversing the removal process above, we can also add eigenvalues to the spectrum and increase the length of the granular chain, as shown in the following algorithm:
\begin{enumerate}
\item Start from a granular chain with $\{r_j\}_{1\leq j\leq n}$ and $\{\gamma_j\}_{2\leq j\leq n}$, whose linearization around the zero solution corresponds to matrix $H$. Suppose the spectrum $H$ is $\Lambda=\{ \lambda_1 > \lambda_2 > ... > \lambda_n >0 \}$ and its image point in $S_{>0}^{n-1}$ is $\boldsymbol{x}=(x_1,x_2,...,x_n)$.

\item For $ \lambda_0 \in (\lambda_{k},\lambda_{k-1}) $ (it is also okay if $\lambda_0 >\lambda_1$ or $\lambda_0<\lambda_n$), we add it to the spectrum $\Lambda$ and put it right between $\lambda_{k-1}$ and $\lambda_k$. The newly obtained spectrum is denoted by $\tilde{\Lambda}$. We extend $\boldsymbol{x}$ to $\tilde{\boldsymbol{x}}=(\tilde{x}_1,...,\tilde{x}_{k-1}, x_0, \tilde{x}_{k+1}, ... , \tilde{x}_{n+1})$ where $0<x_0\ll 1$ and find its associated matrix $\tilde{H}$ in $T_{ \tilde{\Lambda} }$.

\item We change $\boldsymbol{r}$ to $\tilde{\boldsymbol{r}}=(r_1,...,r_n, r_0)$ and change $\boldsymbol{\gamma}$ to $\tilde{\boldsymbol{\gamma}}=(\gamma_2,...,\gamma_n, \gamma_0)$ where $r_0\gg 1$ and $\gamma_0=\lambda_0\frac{4}{3}\rho\pi r_0^3$. Using $\tilde{\boldsymbol{r}}$ and $\tilde{\boldsymbol{\gamma}}$ as initial guess, we solve for a granular chain of length $n+1$ from $\tilde{H}$.
\end{enumerate}

\subsubsection*{Example 2: removing an eigenfrequency}
In Fig.~\ref{fig4}, we show an example of decoupling an eigenfrequency from the linear counterpart of a granular chain system. Here we start from the granular chain system with
\begin{itemize}
\item radii $r_j=1$ for $1\leq j\leq 6$
\item springs $\gamma_1=0$ and $\gamma_j=1$ for all $2\leq j\leq 6$
\item precompressions satisfying $k_1 d_1^p=k_2 d_2^p=...=k_{7} d_{7}^p=C=1$,
\end{itemize}
and its corresponding matrix in $T_{\Lambda}$ has
\begin{itemize}
\item spectrum $\Lambda: ( 1.5483, 1.3549, 1.0702, 0.7542, 0.4781, 0.2987)$
\item image in $S^{5}_{>0}$: $\boldsymbol{x}=(0.1649, 0.3221, 0.4542, 0.5327, 0.5141, 0.3386)$.
\end{itemize}

Setting $\gamma_1$ fixed and letting $x_2 \to 0$, we remove $\lambda_2=1.3549$ from the spectrum $\Lambda$ while keeping all the other eigenvalues. In the process of removing eigenvalue $\lambda_2$, at first we pick a path in the space $S_{>0}^{5}$ to make $x_2\to 0$, letting $x_j$ for $j\neq 2$ gradually become the same, as reflected in the first $18$ steps in Fig.~\ref{fig4}. At the same time, the pair $(r_6,\gamma_6)$ grow to $\infty$ and then the last bead of the chain is dropped at the $19$-th step, after which the length of the granular chain becomes $5$ and only $5$ eigenfrequencies remain. Although $\gamma_2^{(19)}<0$ is an unphysical value, the ability of moving $\boldsymbol{x}$ in $S_{>0}^{4}$ allows us to arrive at a system with all $\gamma_j$ being nonnegative in the end. In particular, the granular chain at the destination ($38$-th step) has the following parameters:
\begin{itemize}
\item $\tilde{ \boldsymbol{r} } = ( 0.9839, 0.9295, 0.9672, 1.1627, 1.2108)$
\item $\tilde{ \boldsymbol{\gamma} } = (0, 0.6050, 0.6792, 2.0333, 1.0924)$.
\end{itemize}

\begin{figure}[H]
\centering
\begin{tabular}{cc}
\includegraphics[width=6cm]{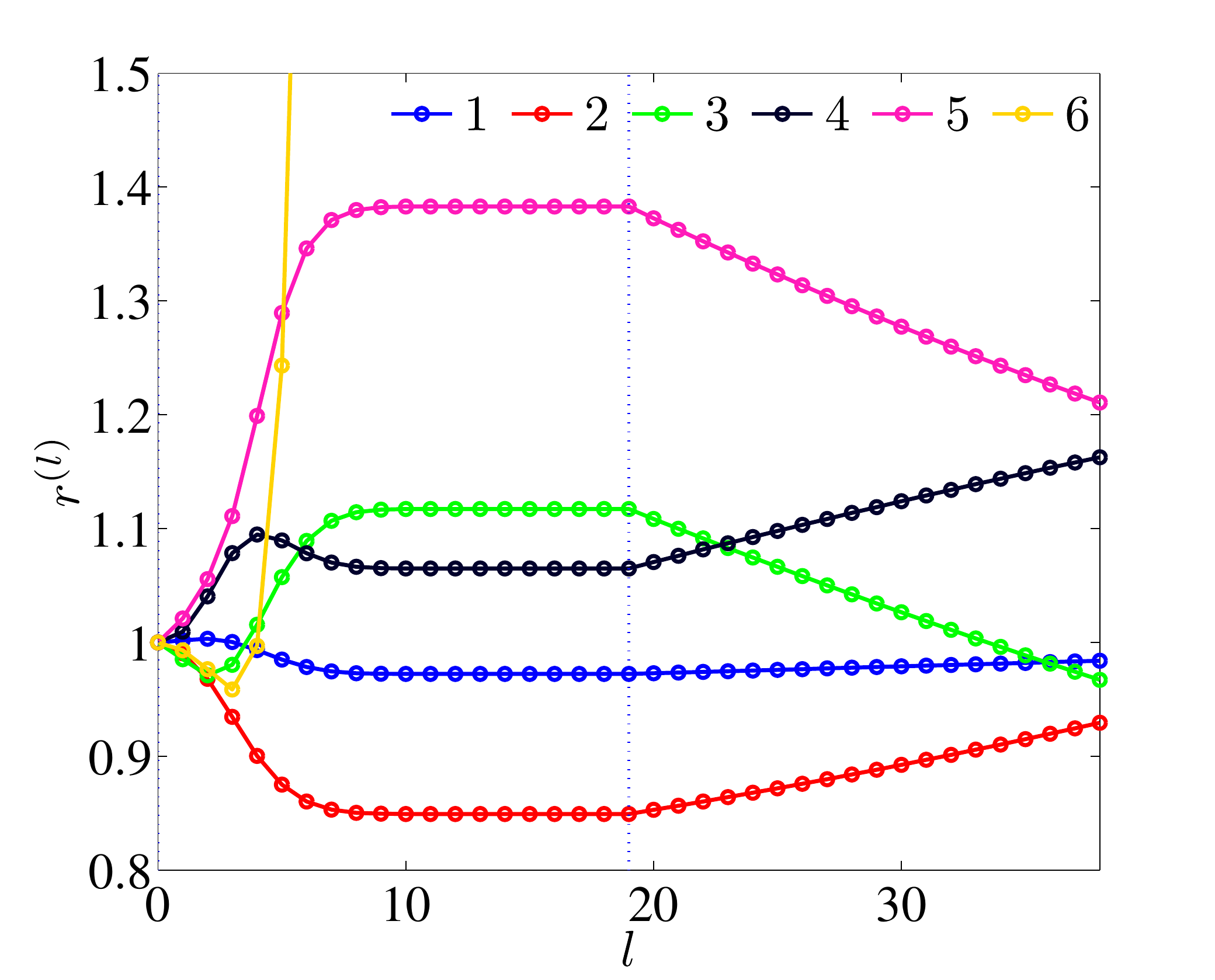}
\includegraphics[width=6cm]{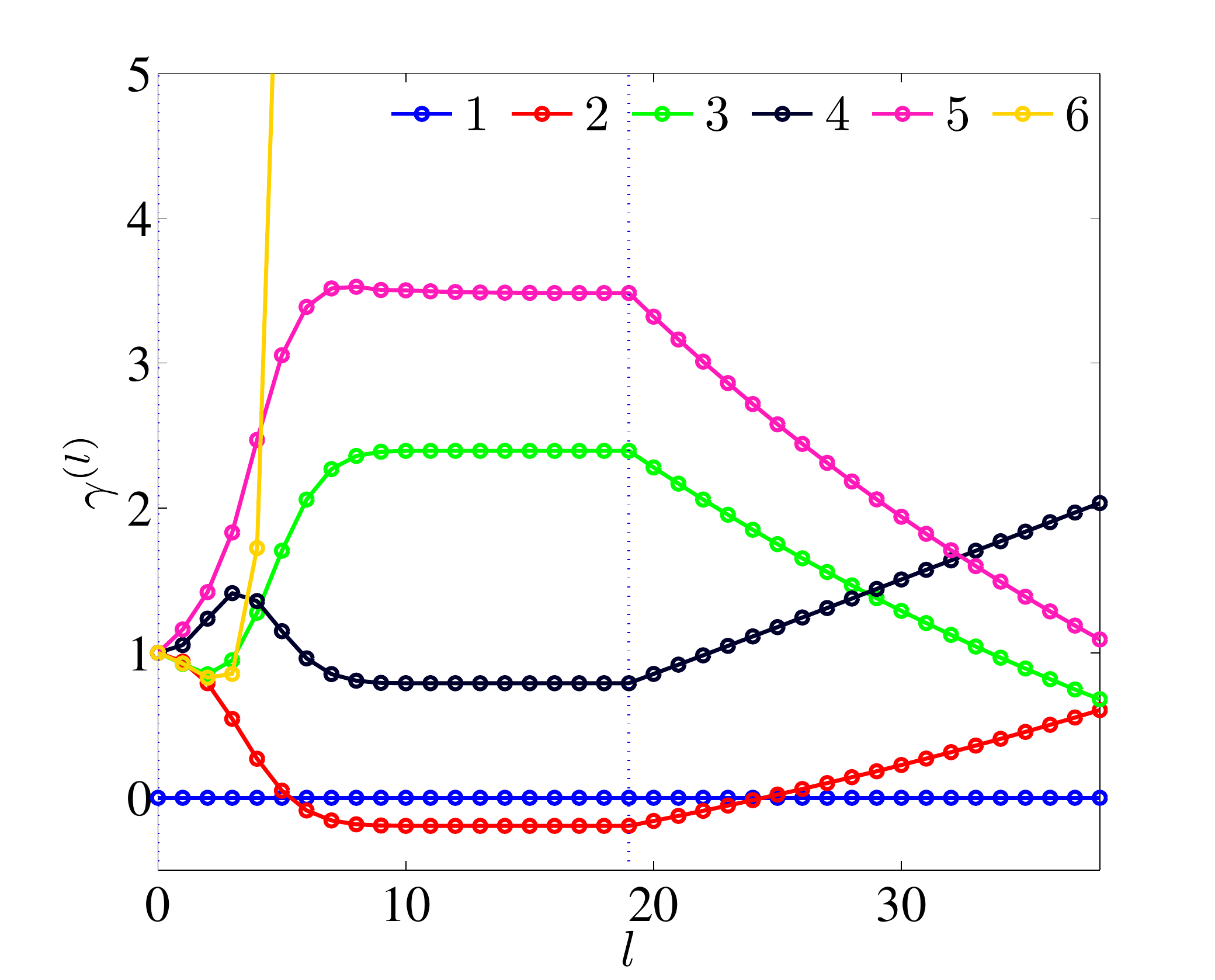}\\
\includegraphics[width=6cm]{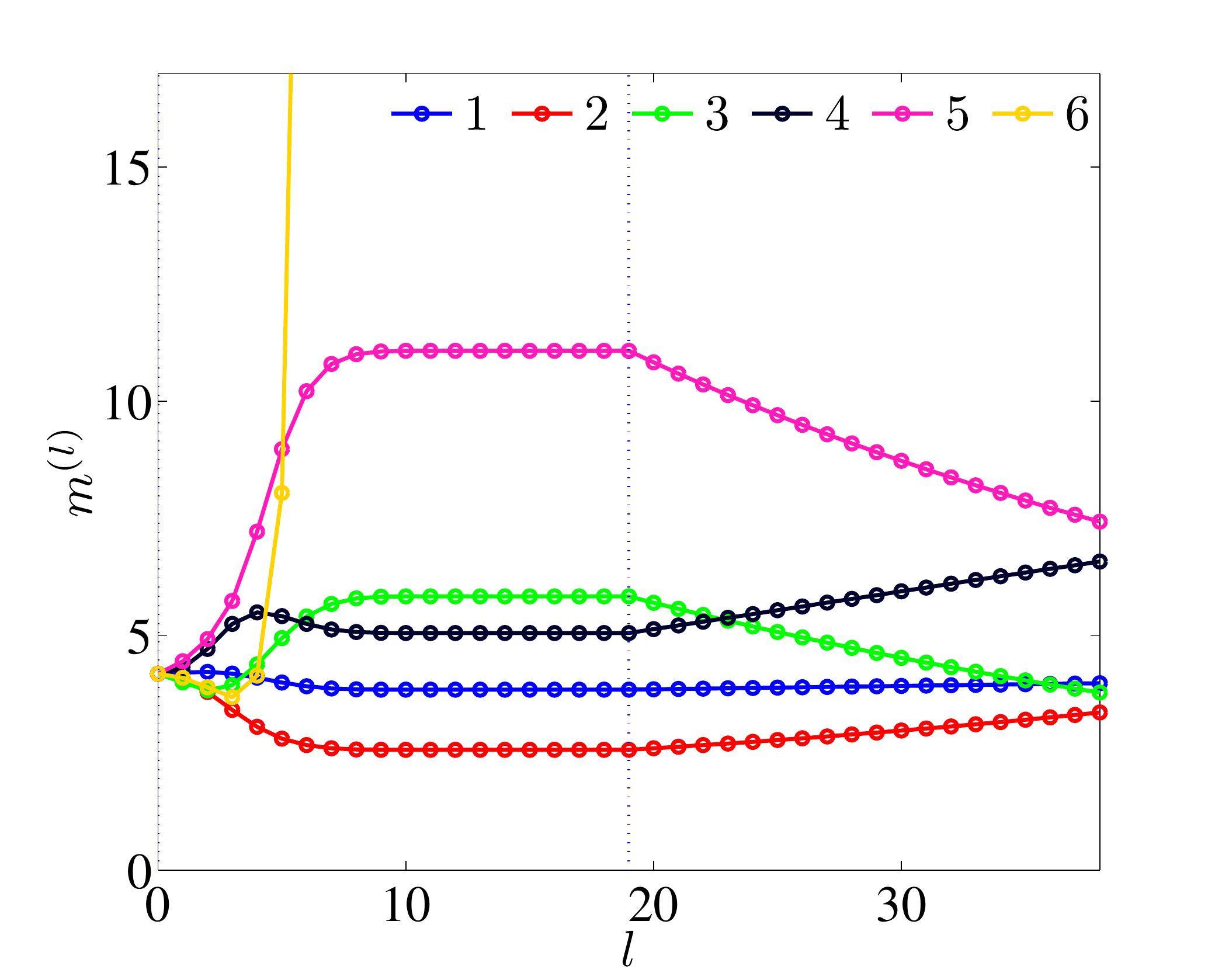}
\includegraphics[width=6cm]{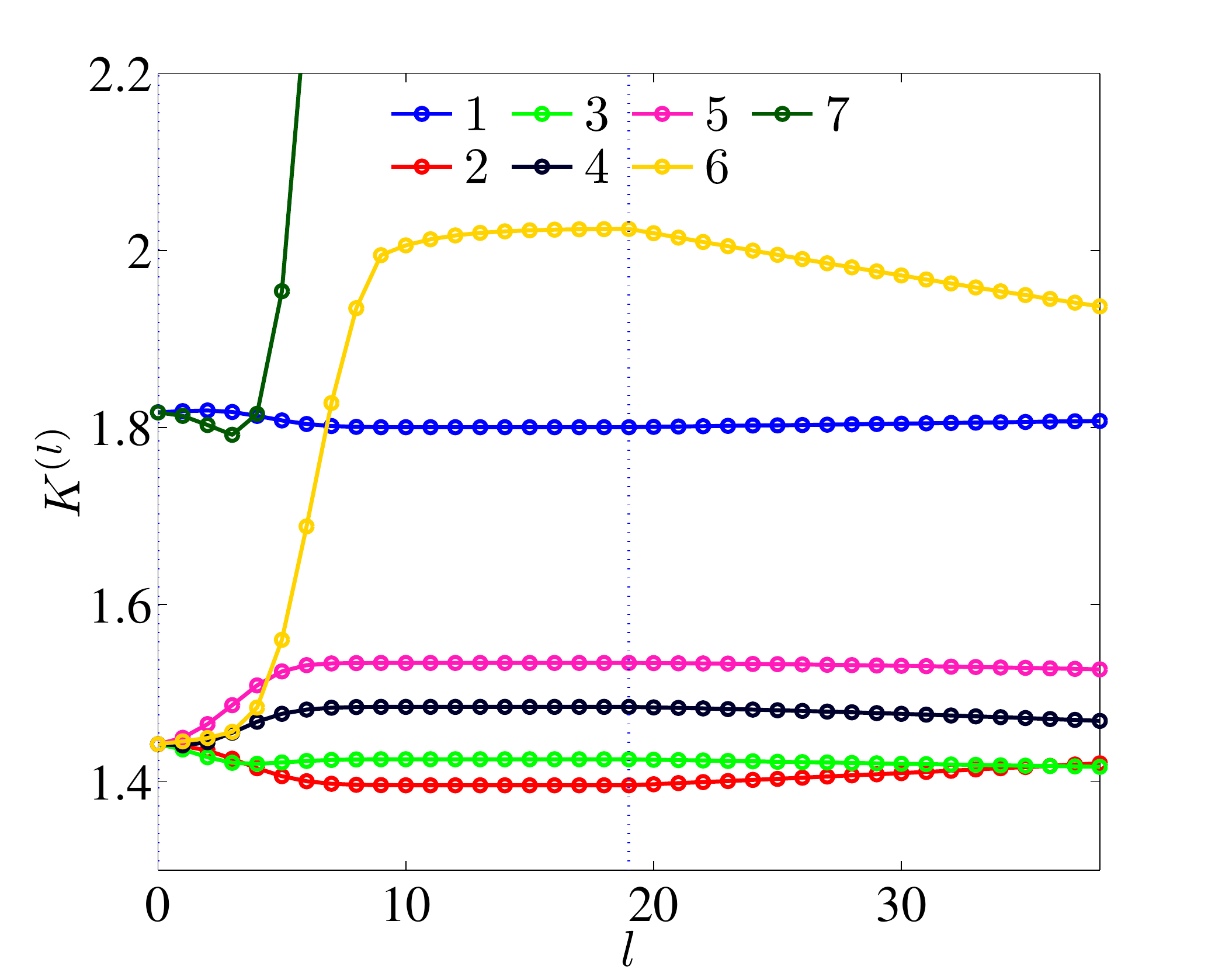}\\
\includegraphics[width=6cm]{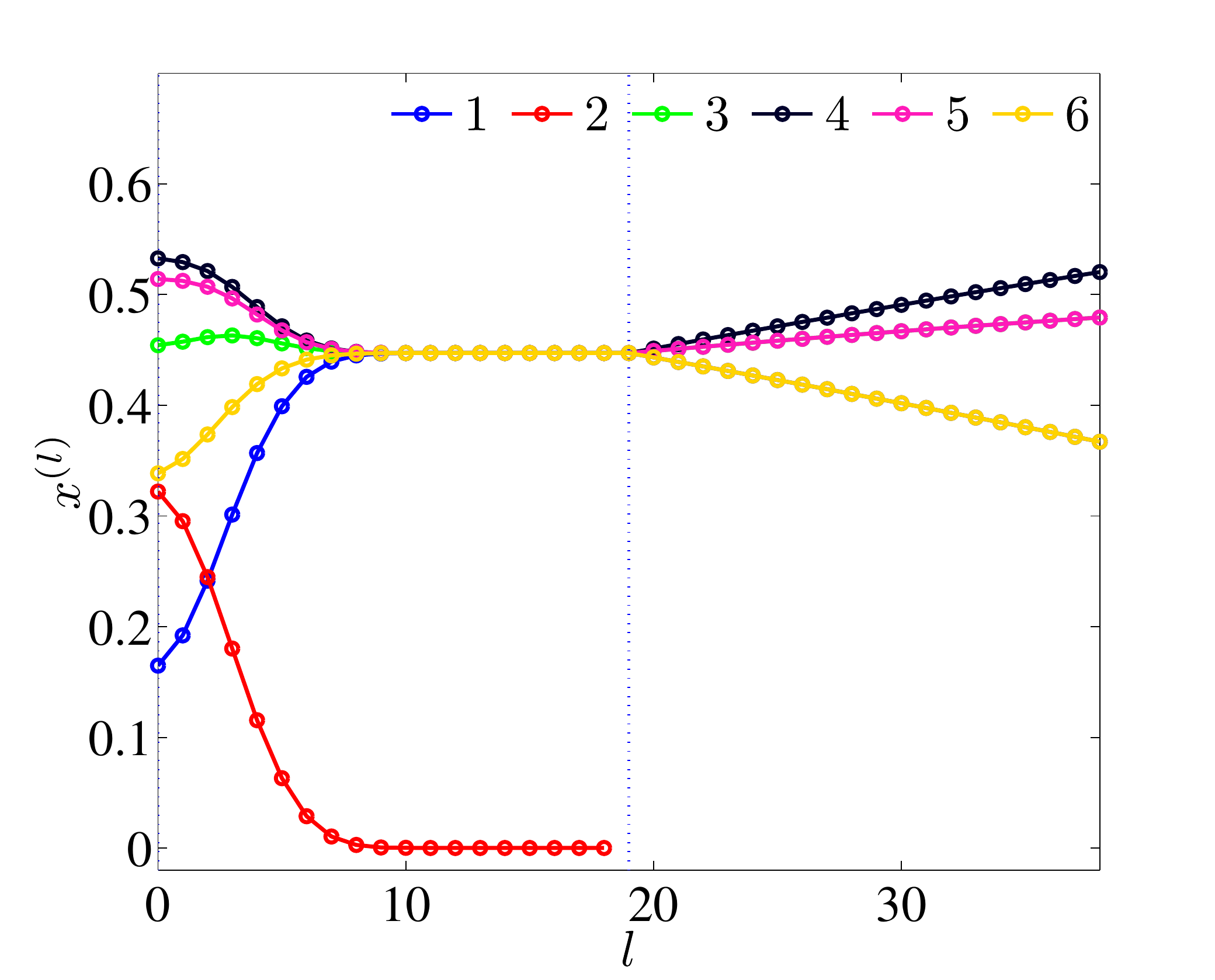}
\includegraphics[width=6cm]{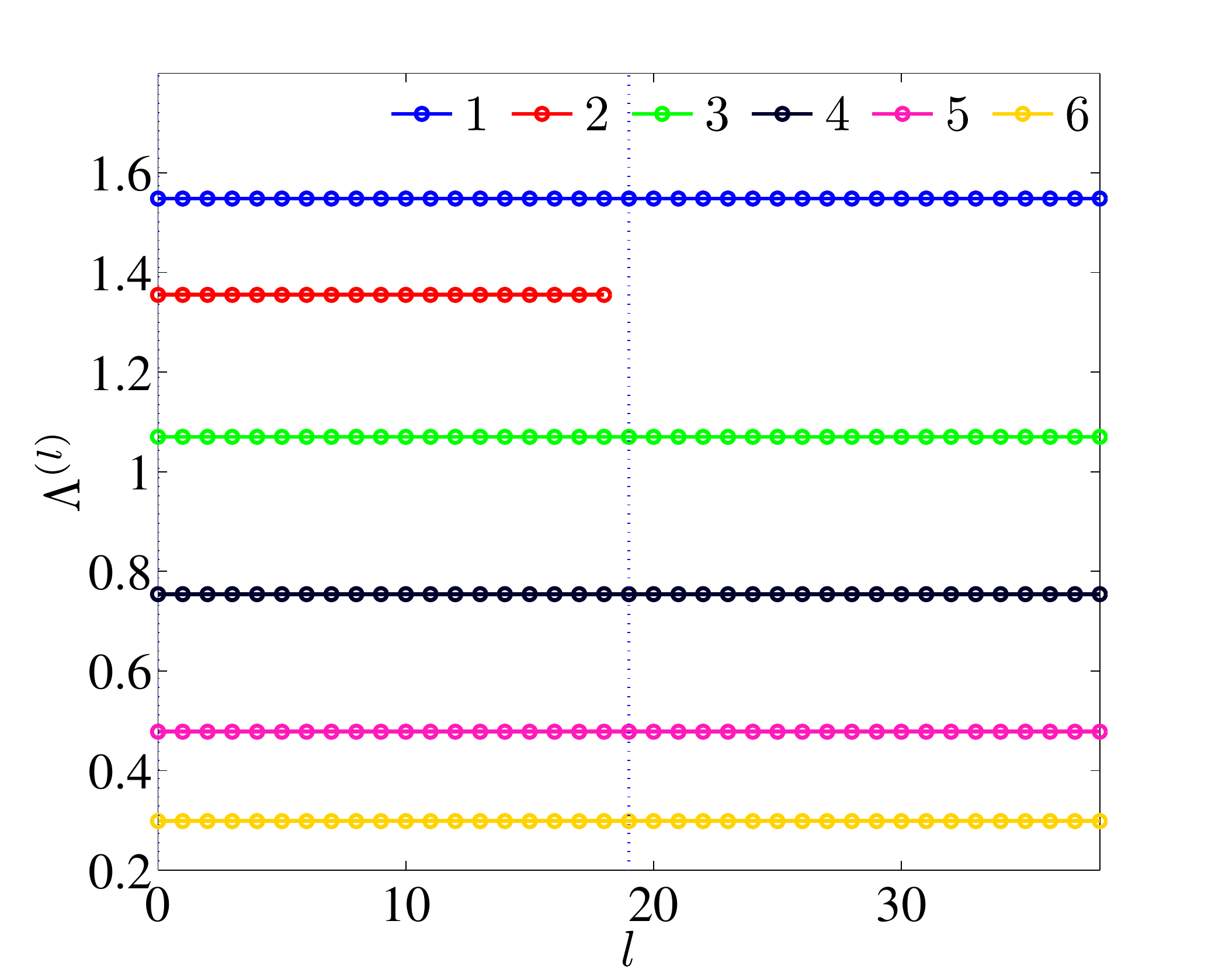}
\end{tabular}
\caption{Top row: radii (left panel) and springs (right panel) in granular chain systems; middle row: masses (left panel) and spring coefficients (right panel) in the linearized systems; bottom: image points (left panel) in the space $S^{n-1}_{>0}$ and the spectrum (right panel). Again, each line with circles stands for a parameter in different isospectral systems. The circles in each column correspond to different parameters in the same system. At step $0$, the granular chain consists of $6$ beads while there are only $5$ after step $19$, which is indicated by the vertical dotted line. It can be clearly seen that we have kept all the eigenvalues except the second largest one. We notice that
  $\gamma_2^{(l)}$ is negative for $6 < l < 24$, yet it becomes
  positive for larger (and smaller) values of $l$, corresponding in the latter
cases to physically relevant systems.}
\label{fig4}
\end{figure}

In pratical, if one wants all $\gamma_j$ to stay nonnegative at every step, then $\boldsymbol{x}^{(18)}$ can be replaced by $(x_1^{(38)}, x_2^{(38)}, x_3^{(38)}, x_4^{(38)}, x_5^{(38)}, x_6^{(38)})$ with $x_2^{(38)}\approx 0$, and a path in $S_{>0}^5$ should be carefully chosen to connect $\boldsymbol{x}^{(0)}$ and the new $\boldsymbol{x}^{(18)}$ such that $\boldsymbol{\gamma}^{(l)}$ never become negative on the path.

\subsubsection*{Example 3: adding an eigenfrequency}
In Fig.~\ref{fig5}, we reverse the process in Example 2 to
showcase an example of the adding algorithm. To be more specific, we start with the granular chain with
\begin{itemize}
\item radii $\tilde{ \boldsymbol{r} } = ( 0.9839, 0.9295, 0.9672, 1.1627, 1.2108)$
\item springs $\tilde{ \boldsymbol{\gamma} } = (0, 0.6050, 0.6792, 2.0333, 1.0924)$
\item spectrum $\Lambda: ( 1.5483, 1.0702, 0.7542, 0.4781, 0.2987)$.
\end{itemize}

Extending the image point in $S^{4}_{>0}$ to $S^{5}_{>0}$ by adding a small number between the second and the third component, we put the eigenvalue $\lambda_2=1.3549$ back to the spectrum and compute the corresponding matrix $H$. At the same time, we add a bead with large $r_6$ and $\gamma_6=\lambda_2\frac{4}{3}\rho\pi r_6^3$ to the end of chain so that it can be a good approximation to the system to be solved from $H$.

As Fig.~\ref{fig5} illustrates, eventually we restore the starting system in Example 2 with
\begin{itemize}
\item radii $r_j=1$ for $1\leq j\leq 6$
\item springs $\gamma_j=1$ for all $2\leq j\leq 6$
\end{itemize}

\begin{figure}[H]
\centering
\begin{tabular}{cc}
\includegraphics[width=6cm]{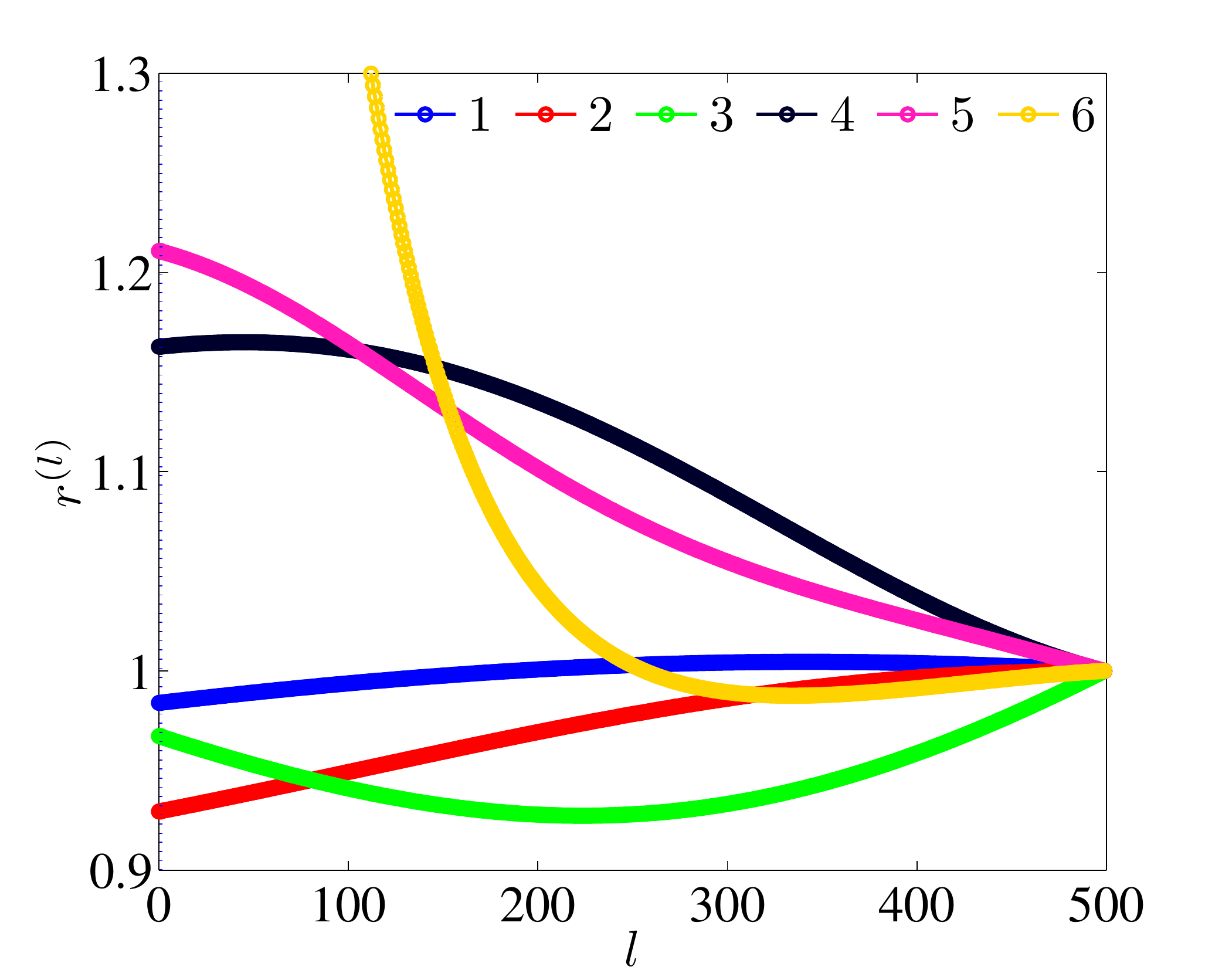}
\includegraphics[width=6cm]{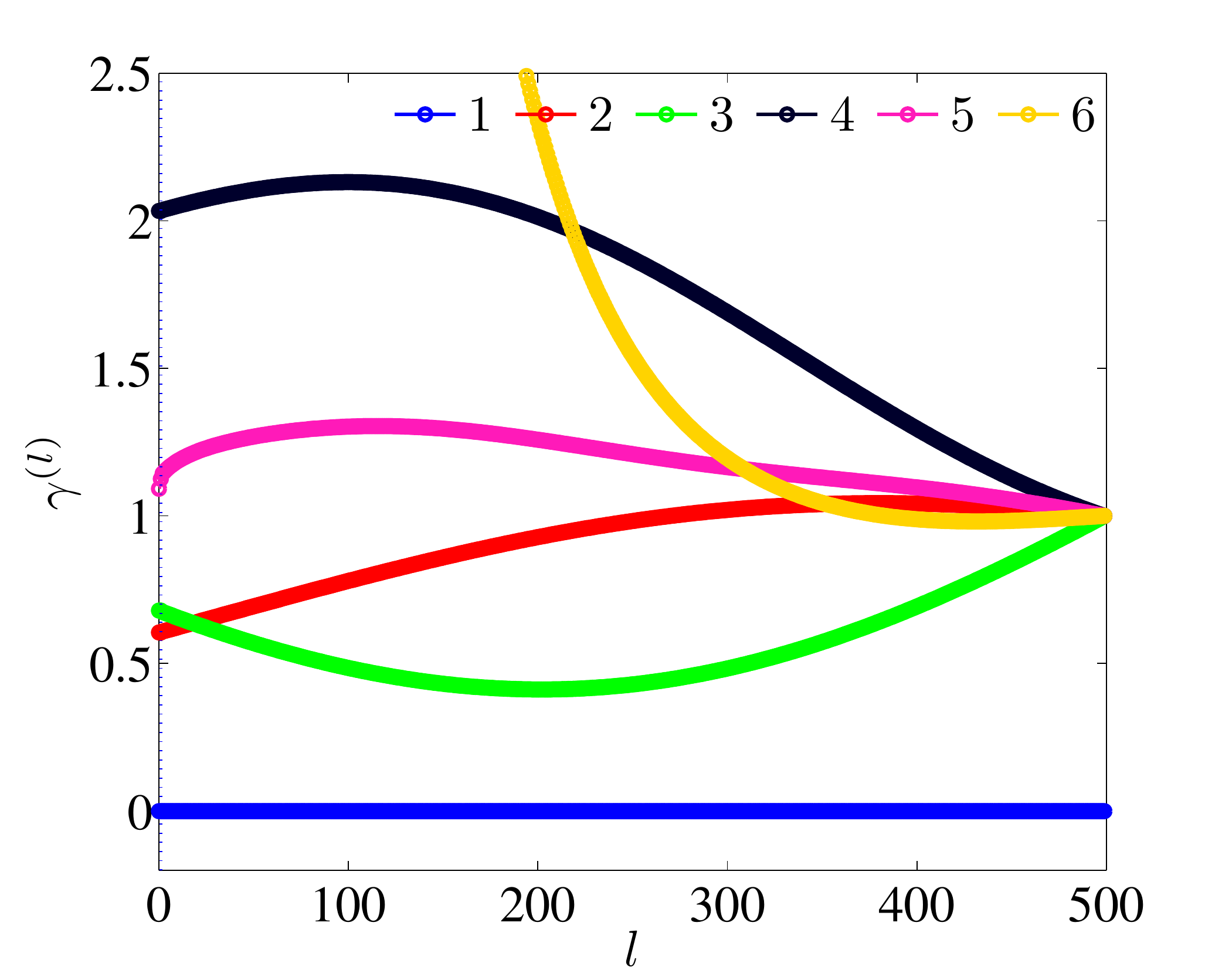}\\
\includegraphics[width=6cm]{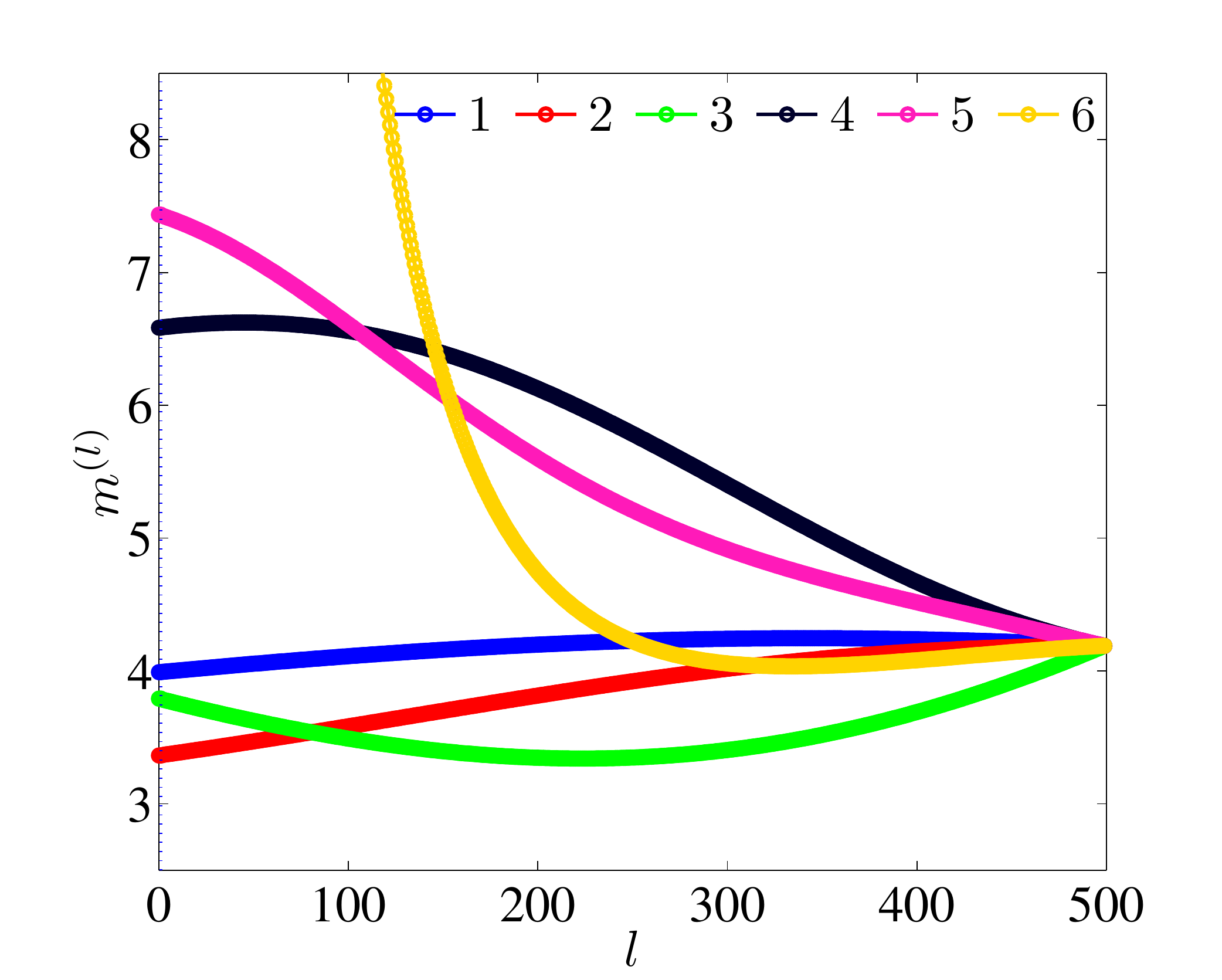}
\includegraphics[width=6cm]{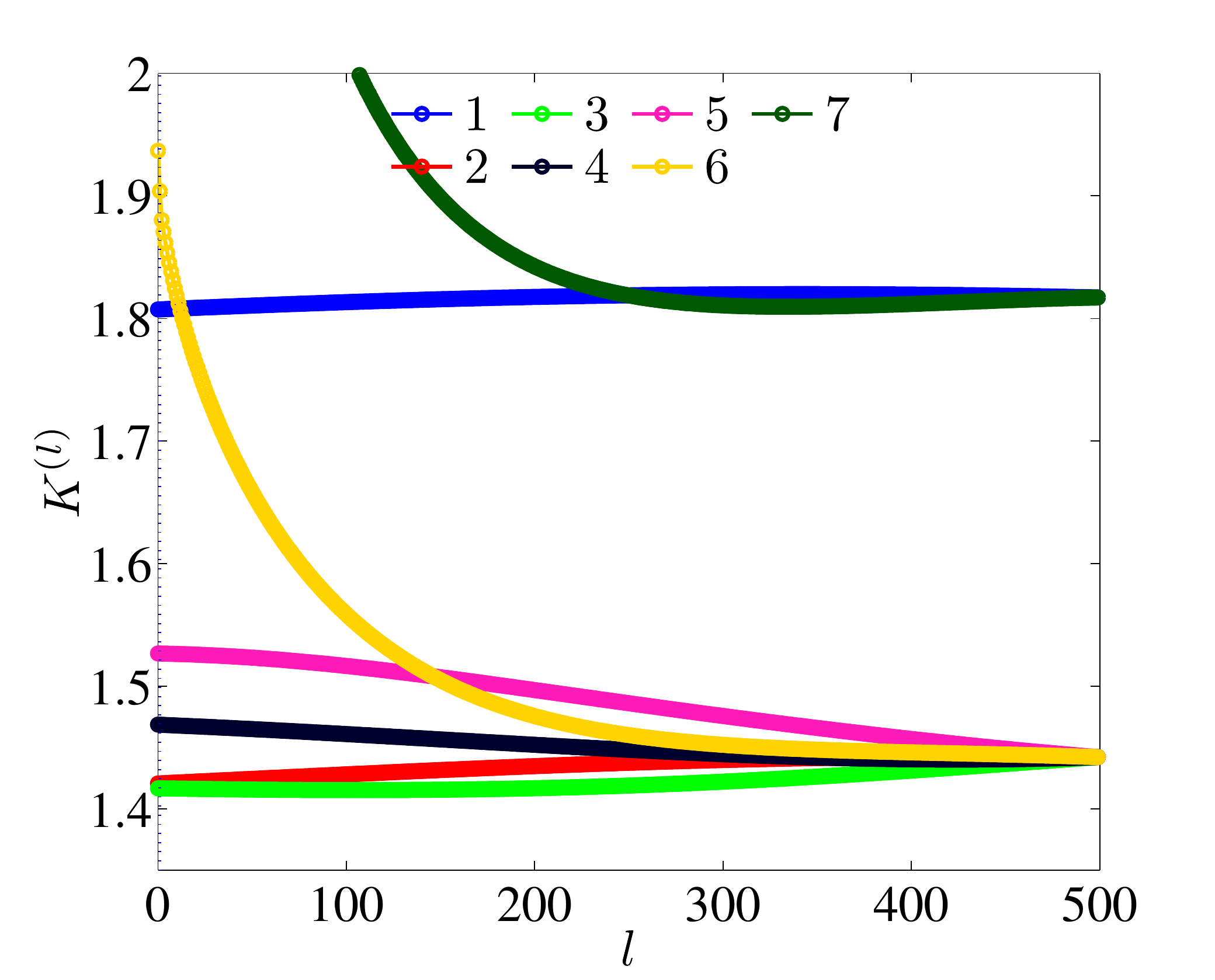}\\
\includegraphics[width=6cm]{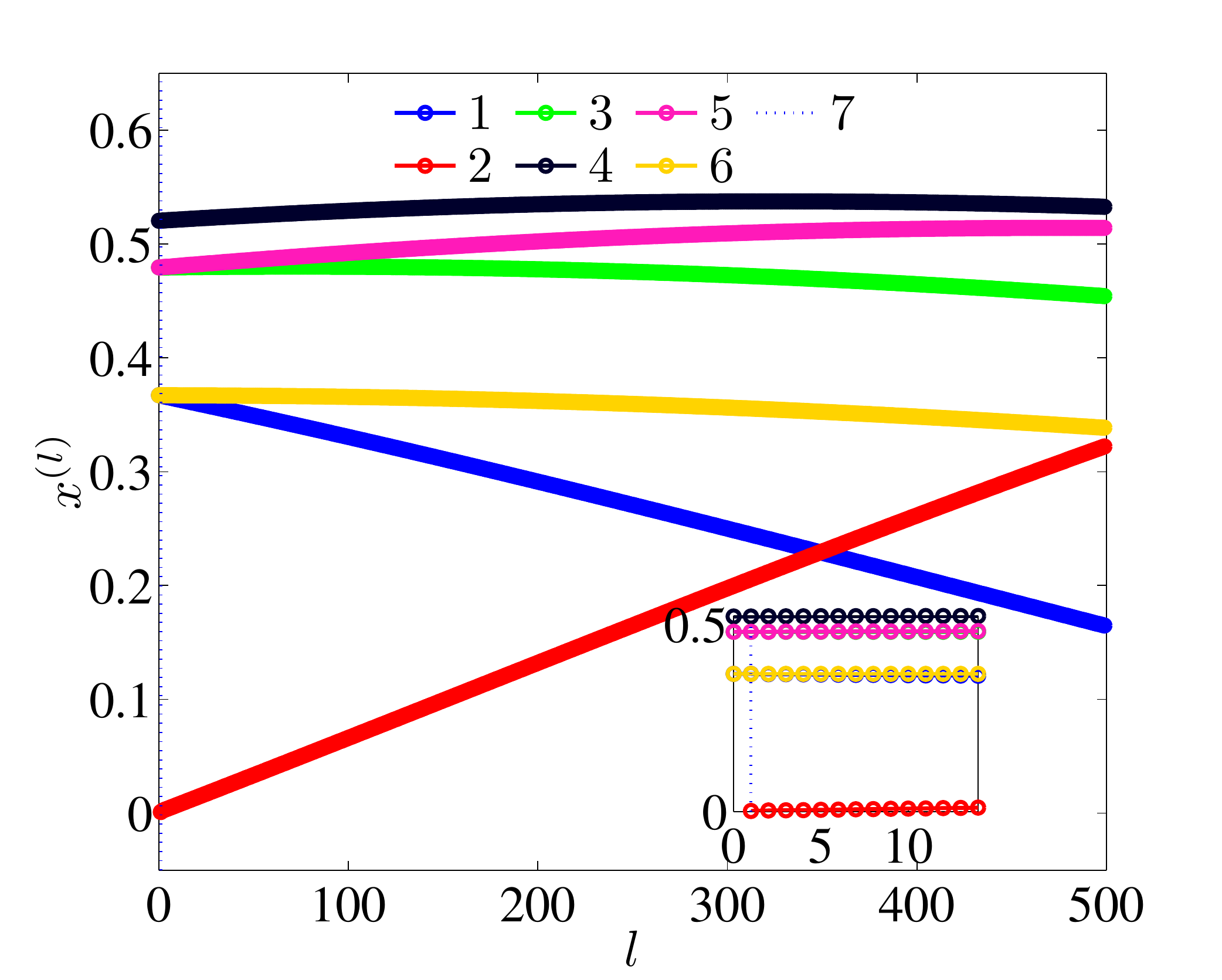}
\includegraphics[width=6cm]{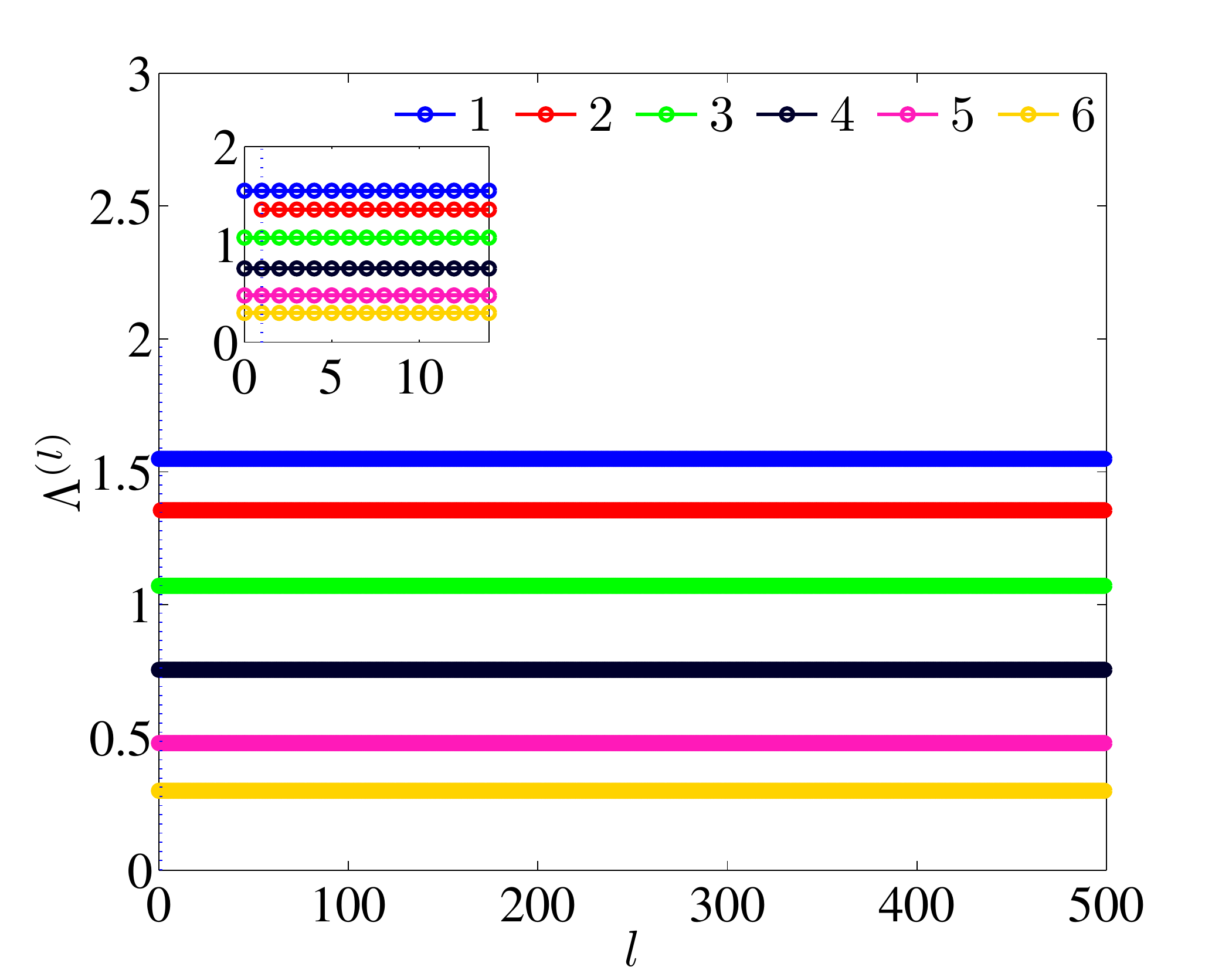}
\end{tabular}
\caption{Top row: radii and springs in granular chain systems; middle row: masses and spring coefficients in the linearized systems; bottom: image point in the space $S^{n-1}_{>0}$ and the spectrum. Similar to the previous
figure, but now for the case of adding an eigenfrequency to the system.}
\label{fig5}
\end{figure}

\section{Wave dynamics in isospectral granular chain}
\subsection{Linear regime}
We now evaluate wave transmission characteristics of isospectral granular chains. Using the state-space approach~\cite{SS}, we calculate the transmission gain as a function of driving frequency $\omega$. Let the dynamic disturbance $F_d$ be applied to the first bead of the system. To measure the force output, $F_N$ at the other end of the system, we write the equations of motion as follows
\begin{eqnarray}
\begin{aligned}
\label{eqn2_12}
\dot{\psi}=A_1\psi + A_2 F_{d}, \
F_N=A_3\psi + A_4F_d
\end{aligned}
\end{eqnarray}

\noindent where,
\[\psi= \begin{pmatrix}
    z_1\\\vdots\\z_n\\\dot{z}_1\\\vdots\\\dot{z}_n
\end{pmatrix};
A_1=\begin{pmatrix}
0 & I\\ M^{-1}B & 0
\end{pmatrix};
A_2=\begin{pmatrix}
0\\\vdots\\0\\1/m_1\\\vdots\\0
\end{pmatrix};
A_3=\begin{pmatrix}
0\\\vdots\\K_{n+1}\\0\\\vdots\\0
\end{pmatrix}^{T};
A_4=0 \]

For the isospectral granular chains, NS-1 and NS-2, discussed in Example 1, we linearize their dynamics. Then we solve the equation (\ref{eqn2_12}) using the Bode function in MATLAB to obtain transmission gain ($F_N/F_d$) as shown in Fig.~\ref{fig6}. We observe that the transmission characteristics of both resulting linearized systems are very similar. This is an interesting observation as it indicates that an \textit{ordered} chain (NS-1) 
can present very similar response characteristics to an apparently \textit{disordered} chain (NS-2) in terms of transmission gain for an elastic wave through the structure. Therefore, this opens up the possibility of constructing a whole  family of \textit{disordered} granular systems with similar wave transmission characteristics, as was theoretically explained through the method of
the previous section.
\begin{figure}[H]
\centering
\includegraphics[width=12cm]{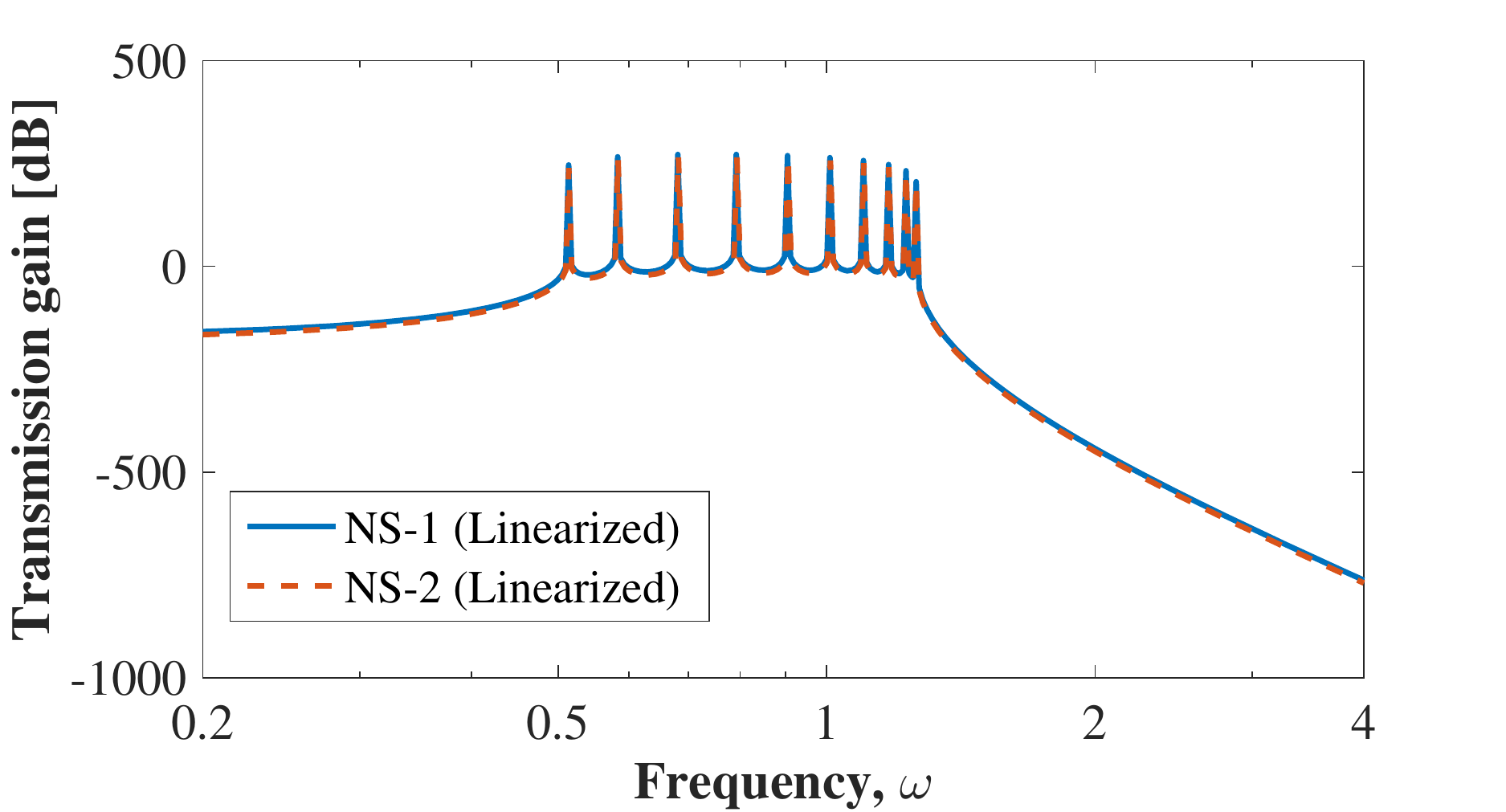}
\caption{Similar transmission gains for two isospectral granular chains, ordered  NS-1 (left) and disordered NS-2 (right), when the dynamics is linearized.}
\label{fig6}
\end{figure}

%

\subsection{Nonlinear regime}

Here we consider the nonlinear dynamics of two granular chain systems (NS-1) and (NS-2) whose linearized models, denoted by (S-1) and (S-2), are isospectral. In order to study the dynamics of a linear system such as (S-1) or (S-2), it suffices to consider its eigenmodes $\{ \boldsymbol{v}_j \cos(\omega_j t) \}_{1\leq j\leq n}$ where $G\boldsymbol{v}_j=M^{1/2}\boldsymbol{v}_j$ is an eigenvector of $H$ with the eigenvalue $\omega_j^2$. Since any initial position of the system can be decomposed as some combination of the eigenmodes and each eigenmode solves the system individually, the exact evolution of the system at any time is available. If the two systems (S-1) and (S-2) are isospectral, their corresponding matrices $H$ will have probably different eigenvectors but the same eigenvalues. Thus, eigenmodes of these two systems may have different profiles but the same frequencies. Suppose a granular chain starts with an initial state that has amplitudes small enough at all bead locations. Then, the evolution up to a finite time can be approximated by that of its linearized model, which is explicitly solved by its eigenmodes and decomposed
in the form of a linear superposition. As we increase the strength of the initial perturbation, the nonlinearity gradually becomes apparent in both systems and the dynamics of the two isospectral granular chains become different.

\begin{figure}[t]
\centering
\begin{tabular}{cc}
\includegraphics[width=7cm]{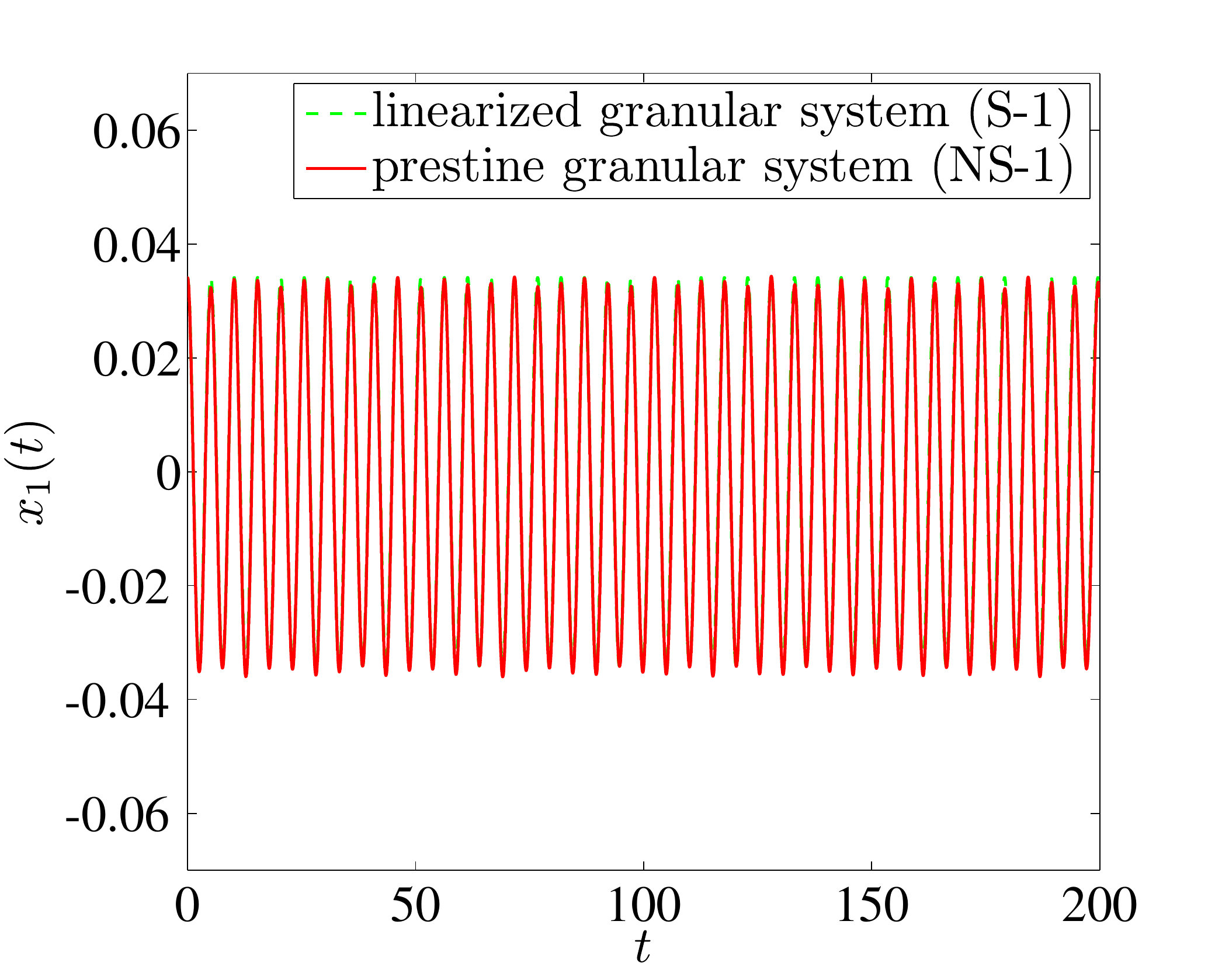}
\includegraphics[width=7cm]{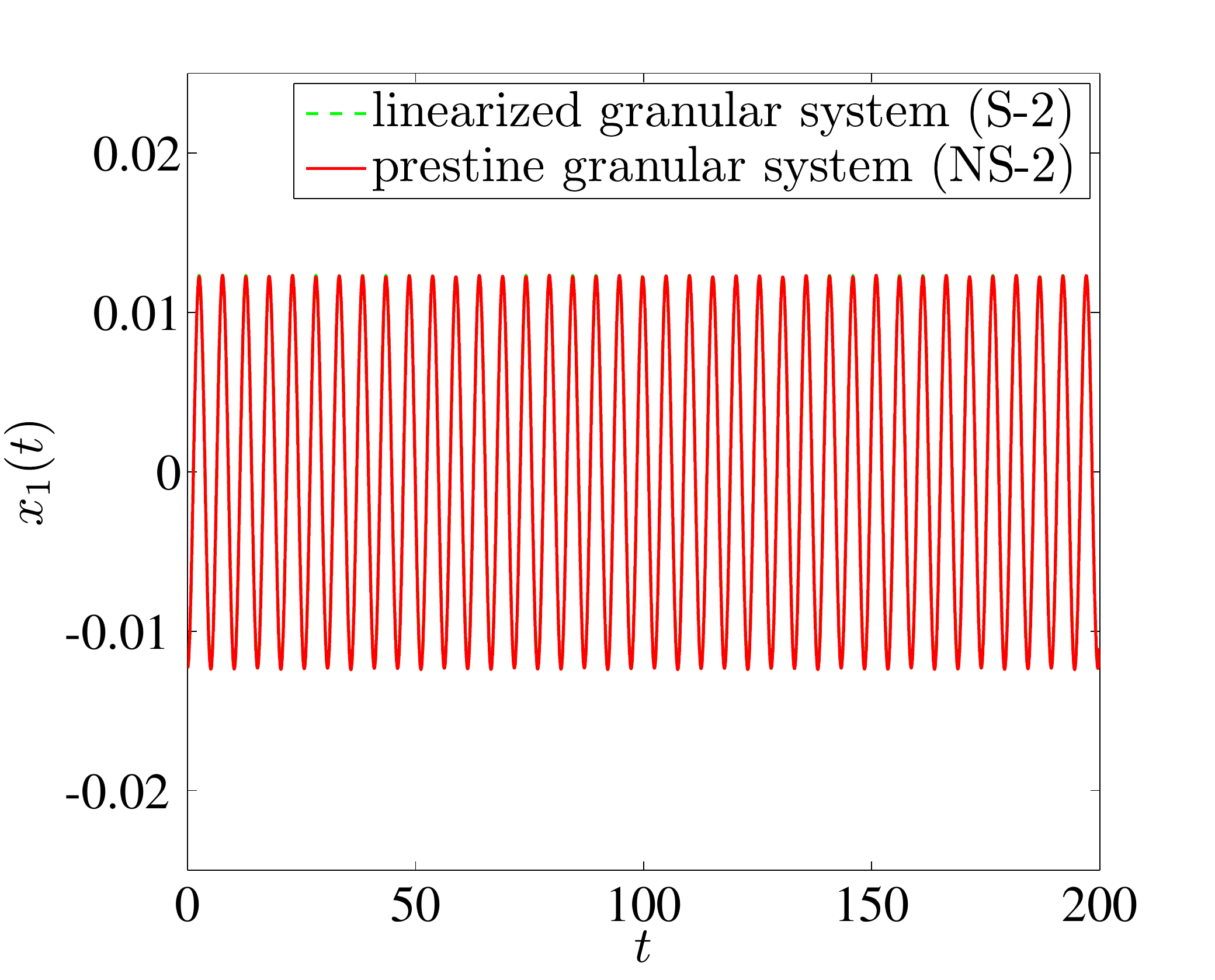}\\
\includegraphics[width=7cm]{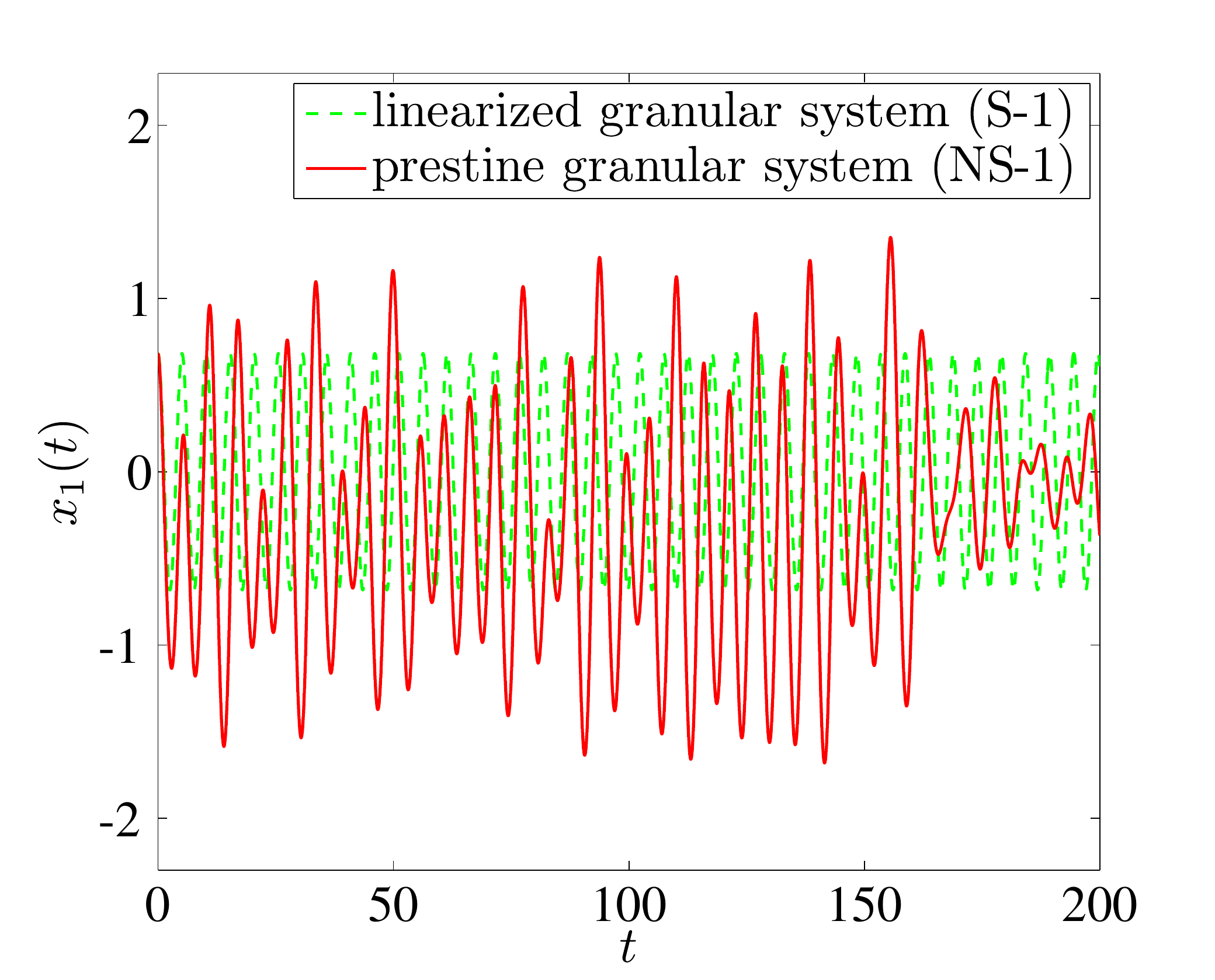}
\includegraphics[width=7cm]{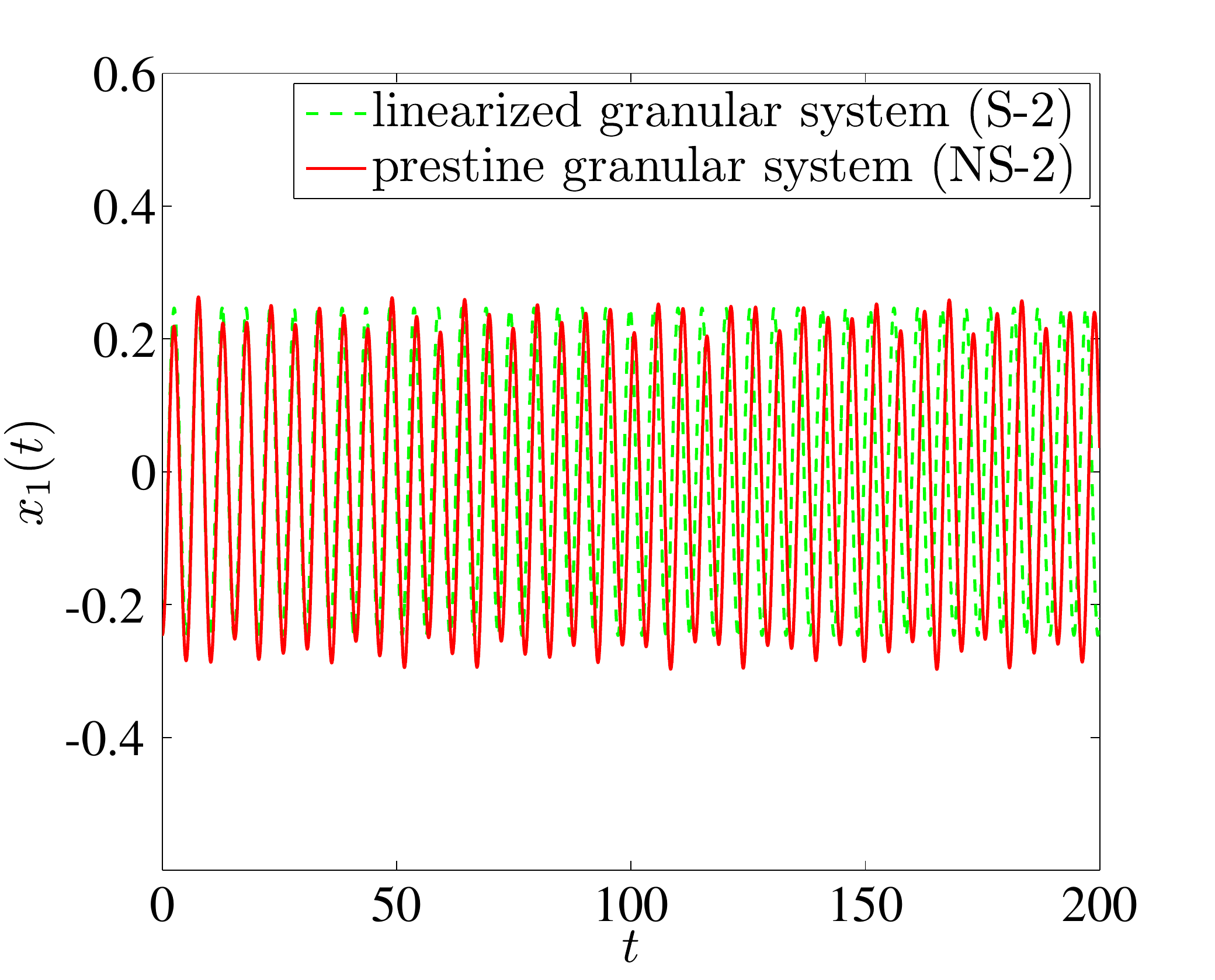}
\end{tabular}
\caption{In the top left (right) panel, we show the time evolution of the first bead in the granular chain NS-1 (NS-2) with a initial perturbation of strength $\delta=0.1$. In the bottom panels, the initial perturbation has been increased such that $\delta=2$, clearly showcasing the differences in the 
anharmonic vibration between the linearly isospectral chains. The dotted line illustrates the evolution of the same initial perturbation in the linearized system. }
\label{fig8}
\end{figure}

In Fig.~\ref{fig8}, we consider two generalized granular chains (NS-1) and (NS-2), which are just the starting and resulting systems in Example 1. By construction, these two systems have the same eigenfrequencies $$\boldsymbol{\omega}=\{ 1.2603, 1.2276, 1.1737, 1.0999, 1.0087, 0.9042, 0.7922, 0.6812, 0.5833, 0.5140\}$$ for their linearized models (S-1) and (S-2).
If we excite each of these two systems with a small perturbation that has the form $\delta \boldsymbol{v}_2$ where $\boldsymbol{v}_2$ is the eigenmode for $\omega_2={1.2276}$ and $\|\boldsymbol{v}_2\|_2=1$, the system will almost follow that eigenmode to evolve periodically. This trend can be seen in the top panels of Fig.~\ref{fig8} when $\delta = 0.1 \ll 1$. As the perturbation becomes stronger, the time evolution of each system starts to deviate from a single
eigenmode (or more generally a superposition of eigenmodes)
in terms of both its period and amplitude. We can observe this phenomenon in the bottom panels of Fig.~\ref{fig8}, where $\delta = 2$.

For a more systematic study of the perturbation in $\delta$, we seek time-periodic solutions to equation~(\ref{eqn1_gamma}) with frequency $\omega$ and use the new time variable $\tau=\omega t$, making equation~(\ref{eqn1_gamma}) look like:
\begin{eqnarray}
\label{eqn_expand_0}
m_j \omega^2 \frac{\partial^2}{\partial \tau^2}z_j=k_j (z_{j-1}-z_{j}+d_j)_+^p - k_{j+1} (z_{j}-z_{j+1}+d_{j+1})_+^p-\gamma_j z_j
\end{eqnarray}
Then we consider the expansions of the frequency $\omega=\omega^{(1)}+\delta\omega^{(2)}+\delta^2\omega^{(3)}+O(\delta^3)$ and the solution $z_j=\delta z_j^{(1)}+\delta^2 z_j^{(2)}+\delta^3 z_j^{(3)}+O(\delta^4)$ in small parameter $\delta$. Substitute these expressions into equation~(\ref{eqn_expand_0}) and combining terms with the same order to obtain:
\begin{eqnarray}
\label{eqn_expand_1}
m_j (\omega^{(1)})^2 \frac{\partial^2}{\partial \tau^2}z_j^{(1)}=K_j (z_{j-1}^{(1)}-z_j^{(1)}) - K_{j+1} (z_{j}^{(1)}-z_{j+1}^{(1)})-\gamma_j z_j^{(1)} ,\\
\label{eqn_expand_2}
\begin{split}
m_j ( (\omega^{(1)})^2 \frac{\partial^2}{\partial \tau^2}z_j^{(2)} +2 \omega^{(1)} \omega^{(2)} \frac{\partial^2}{\partial \tau^2}z_j^{(1)} ) &=& K_j (z_{j-1}^{(2)}-z_j^{(2)}) - K_{j+1} (z_{j}^{(2)}-z_{j+1}^{(2)})-\gamma_j z_j^{(2)} \\
&  &+ Q_j (z_{j-1}^{(1)}-z_j^{(1)})^2- Q_{j+1} (z_{j}^{(1)}-z_{j+1}^{(1)})^2
\end{split}
\end{eqnarray}
where $Q_j=\frac{1}{2}p(p-1)k_j d_j^{p-2}$. The limiting frequency $\omega^{(1)}$ in equation~(\ref{eqn_expand_1}) is what we have computed earlier by solving equation~(\ref{eqn1_ln_matrix}). Here we write $\boldsymbol{z}^{(1)}=\boldsymbol{Z}^{(1)}e^{i\tau}+\overline{\boldsymbol{Z}^{(1)}}e^{-i\tau}$ with $\boldsymbol{Z}^{(1)}$ being the eigenmode for $\omega^{(1)}$ in equation~(\ref{eqn1_gamma}). 

Being interested in the leading-order corrections $\omega^{(2)}$ and $z_j^{(2)}$, we rewrite equation~(\ref{eqn_expand_2}) as follows:
\begin{eqnarray}
\label{eqn_expand_2_2}
\begin{split}
&  &-m_j (\omega^{(1)})^2 \frac{\partial^2}{\partial \tau^2}z_j^{(2)} +K_j (z_{j-1}^{(2)}-z_j^{(2)}) - K_{j+1} (z_{j}^{(2)}-z_{j+1}^{(2)})-\gamma_j z_j^{(2)} \\
&=&2 \omega^{(1)} \omega^{(2)} m_j \frac{\partial^2}{\partial \tau^2}z_j^{(1)} - [ Q_j (z_{j-1}^{(1)}-z_j^{(1)})^2- Q_{j+1} (z_{j}^{(1)}-z_{j+1}^{(1)})^2]
\end{split}
\end{eqnarray}
where the right hand side has $e^{i\tau}$ and $e^{2i\tau}$ terms. By projecting equation~(\ref{eqn_expand_2_2}) to the span of $\boldsymbol{z}^{(1)}$, we notice $2 \omega^{(1)} \omega^{(2)} \langle  \boldsymbol{z}^{(1)}, M \boldsymbol{z}^{(1)} \rangle = 0$ hence $\omega^{(2)}=0$.
Using the ansatz $\boldsymbol{z}^{(2)}=\boldsymbol{Z}^{(2)}e^{2i\tau}+\overline{\boldsymbol{Z}^{(2)}}e^{-2i\tau}+\boldsymbol{Y}^{(2)}$, we obtain
\begin{eqnarray}
\label{eqn_expand_2_3}
m_j \omega_1^2 4 Z_j^{(2)} +K_j (Z_{j-1}^{(2)}-Z_j^{(2)}) - K_{j+1} (Z_{j}^{(2)}-Z_{j+1}^{(2)})-\gamma_j Z_j^{(2)} 
&=&-Q_j (Z_{j-1}^{(1)}-Z_j^{(1)})^2 + Q_{j+1} (Z_{j}^{(1)}-Z_{j+1}^{(1)})^2 \\
\label{eqn_expand_2_4}
K_j (Y_{j-1}^{(2)}-Y_j^{(2)}) - K_{j+1} (Y_{j}^{(2)}-Y_{j+1}^{(2)})-\gamma_j Y_j^{(2)} 
&=&-2Q_j |Z_{j-1}^{(1)}-Z_j^{(1)}|^2 + 2Q_{j+1} |Z_{j}^{(1)}-Z_{j+1}^{(1)}|^2
\end{eqnarray}
where $\boldsymbol{Z}^{(2)}$ and $\boldsymbol{Y}^{(2)}$  can be solved from equation~(\ref{eqn_expand_2_3}) and equation~(\ref{eqn_expand_2_4}), respectively. Moreover, assuming $\boldsymbol{z}^{(3)}=\boldsymbol{Z}^{(3)}e^{3i\tau}+\overline{\boldsymbol{Z}^{(3)}}e^{-3i\tau}+\boldsymbol{Y}^{(1)}e^{i\tau}+\overline{\boldsymbol{Y}^{(1)}}e^{-i\tau}$ and looking at the $e^{i\tau}$ terms of order $O(\delta^3)$, we have the following equation 
\begin{eqnarray}
\label{eqn_expand_2_5}
\begin{split}
2\omega^{(1)} \omega^{(3)} m_j Z_j^{(1)}&=& -m_j (\omega^{(1)})^2 Y_j^{(1)} -K_j (Y_{j-1}^{(1)}-Y_j^{(1)}) + K_{j+1} (Y_{j}^{(1)}-Y_{j+1}^{(1)})-\gamma_j Y_j^{(1)} \\
&  &- 2 Q_j[ (Z_{j-1}^{(1)}-Z_j^{(1)})(Y_{j-1}^{(2)}-Y_{j}^{(2)}) + (\overline{Z_{j-1}^{(1)}}-\overline{Z_j^{(1)}})(Z_{j-1}^{(2)}-Z_{j}^{(2)}) ] \\
&  &+2 Q_{j+1}[ (Z_{j}^{(1)}-Z_{j+1}^{(1)})(Y_{j}^{(2)}-Y_{j+1}^{(2)}) + (\overline{Z_{j}^{(1)}}-\overline{Z_{j+1}^{(1)}})(Z_{j}^{(2)}-Z_{j+1}^{(2)}) ]\\
&  &-3 P_j  (Z_{j-1}^{(1)}-Z_j^{(1)})|Z_{j-1}^{(1)}-Z_{j}^{(1)}|^2 +3\delta P_{j+1} (Z_{j}^{(1)}-Z_{j+1}^{(1)})|Z_{j}^{(1)}-Z_{j+1}^{(1)}|^2
\end{split}
\end{eqnarray}

where $P_j=\frac{1}{6}p(p-1)(p-2)k_j d_j^{p-3}$.  After projecting equation~(\ref{eqn_expand_2_5}) to the span of $\boldsymbol{Z}^{(1)}$, the terms containing $\boldsymbol{Y}^{(1)}$ will vanish and $\omega^{(3)}$ will be obtained. Since $\omega^{(3)}$ depends on the eigenmode $\boldsymbol{z}^{(1)}$ for $\omega^{(1)}$ and isospectral systems usually have different eigenmodes, the leading-order correction to the eigenfrequency varies in isospectral systems, as illustrated in the left panel of Fig.~\ref{fig_omega3}. This feature quantitatively illustrates
the role of the nonlinearity in leading to deviations between the
frequencies of linearly isospectral systems.

Complementing this theoretical analysis, we numerically computed the time-periodic solutions to equation~(\ref{eqn1_gamma}) as continuations of linear eigenmodes in parameter $\delta$ (or the amplitude of the solution). In pratical, we solved solutions with period $2\pi$ to equation~(\ref{eqn_expand_0}) on a fixed domain $[0,2\pi)$. 
To validate the leading-order approximations of the frequencies for small $\delta$, we compared them with the frequencies of the numerical time-periodic solutions and found a good agreement between them (see the middle and right panels of Fig.~\ref{fig_omega3}). 

\begin{figure}[!htbp]
\centering
\begin{tabular}{ccc}
\includegraphics[width=5cm]{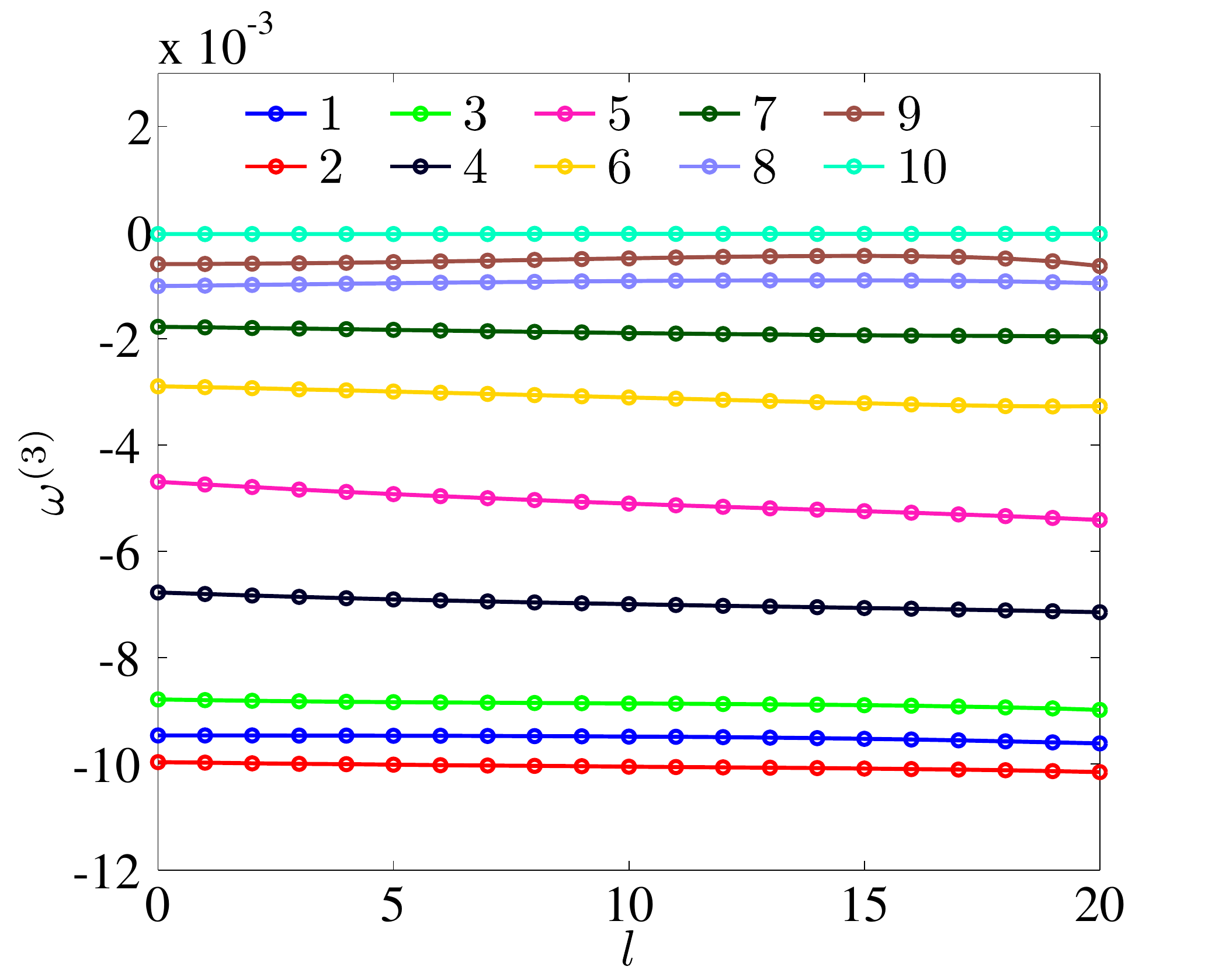}
\includegraphics[width=5cm]{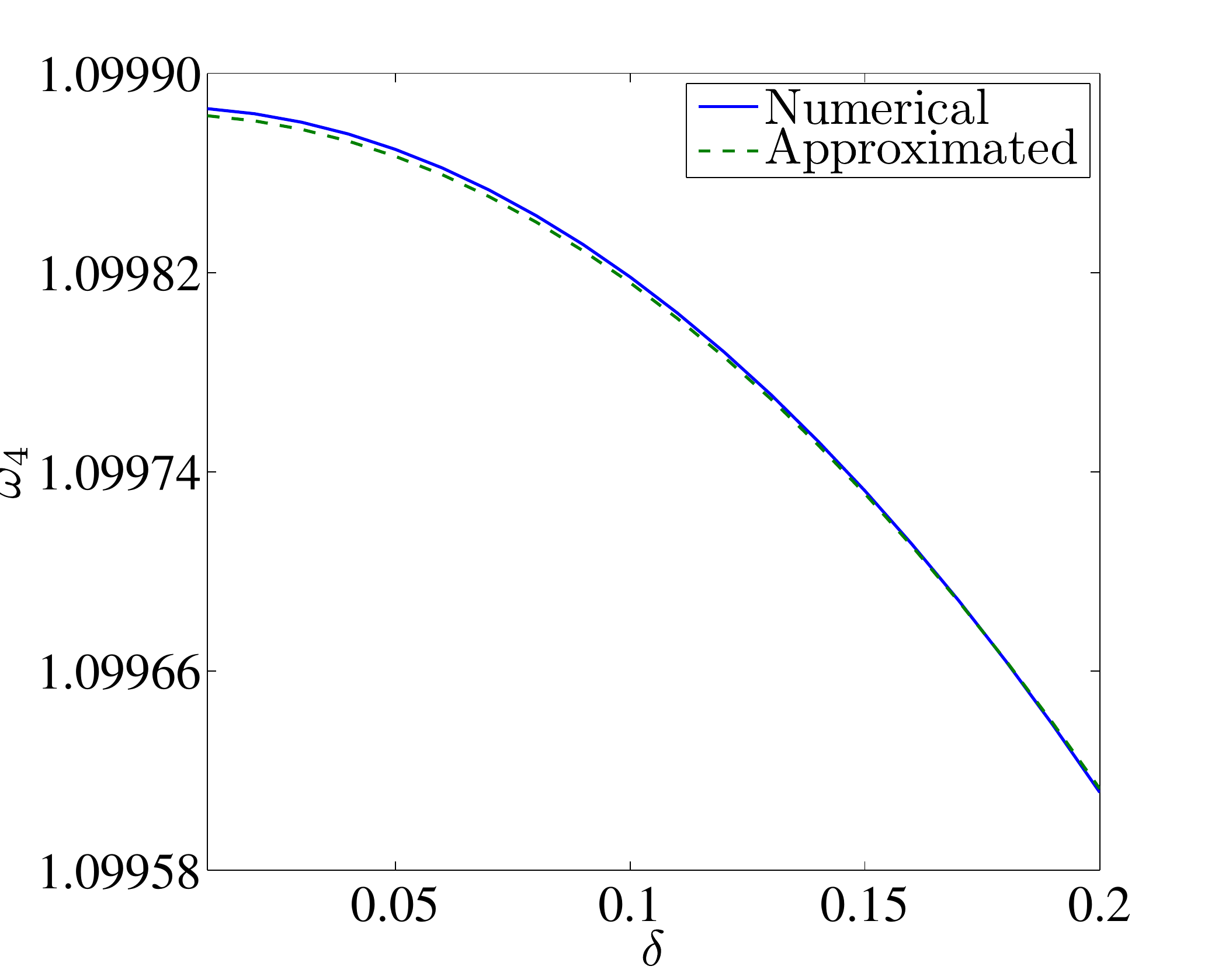}
\includegraphics[width=5cm]{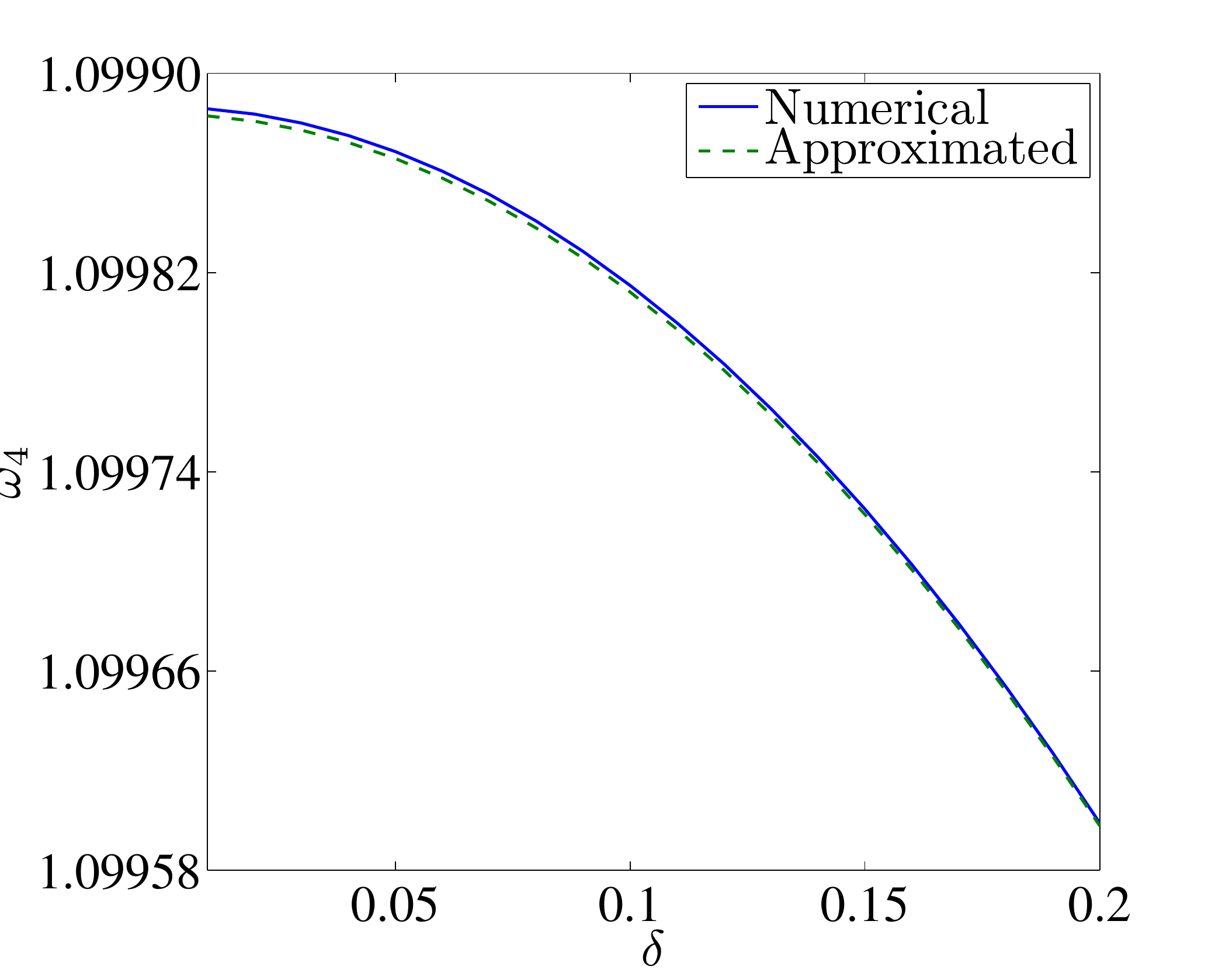}
\end{tabular}
\caption{In the left panel, we illustrate $\omega^{(3)}$, leading-order coefficients, for the frequencies $\omega^{(1)}$ of different isospectral systems in Fig.~\ref{fig2} that are indexed by $l$. In particular, $l=0$ and $l=20$ correspond to granular chains (NS-1) and (NS-2), respectively. The middle and right panels show the changes of the fourth frequency $\omega_4$ (where $\omega_4^{(0)}=1.0999$) versus the growth of $\delta$ for (NS-1) and (NS-2), respectively. The solid line is the frequency of the numerically computed time-periodic solution while the dashed line is the $O(\delta^2)$ approximation of the frequency using $\omega^{(3)}_4$. }
\label{fig_omega3}
\end{figure}

An intriguing observation is that in all the isospectral systems we considered here, the leading-order corrections to the frequencies are nonpositive, i.e., $\omega^{(3)}_j\leq 0$ for $1\leq j\leq n$. Though it is not obvious from the expression of $\omega^{(3)}$, the decreasing of eigenfrequencies has been verified by following the continuation of the eigenmodes in parameter $\delta$, as revealed in Fig.~\ref{fig_omega_numerical}. We note that this is a typical scenario among granular chains since similar features are observed in other granular chains~\cite{Narisetti, Cabaret}.
We attribute this feature to the effectively {\it self-defocusing}
nature of the nonlinearity
(see also, e.g., the relevant asymptotic calculation of~\cite{chong2}
in the simpler case of homogeneous granular crystals).
On the other hand, though the natural frequencies of isospectral granular chains in Fig.~\ref{fig_omega_numerical} tend to decrease differently, breaking down the equivalence inherited from the linear limit, there can also be quantitative similarities in the change of their frequencies. For example, while the isospectral systems in Fig.~\ref{fig2} are quite different in parameters, their changes in eigenfrequencies differ slightly
before $\delta\sim O(1)$ (see Fig.~\ref{fig_omega3} and Fig.~\ref{fig_omega_numerical}). 
Due to the implicit relationship between the parameters of the isospectal map and the eigenmodes of isospectral systems, exploring the possibility of enhancing the connection between isospectral systems in the nonlinear
regime remains a substantial, yet important challenge
worth considering in the future.

%

%
\begin{figure}[!htbp]
\centering
\begin{tabular}{ccc}
\includegraphics[width=5cm]{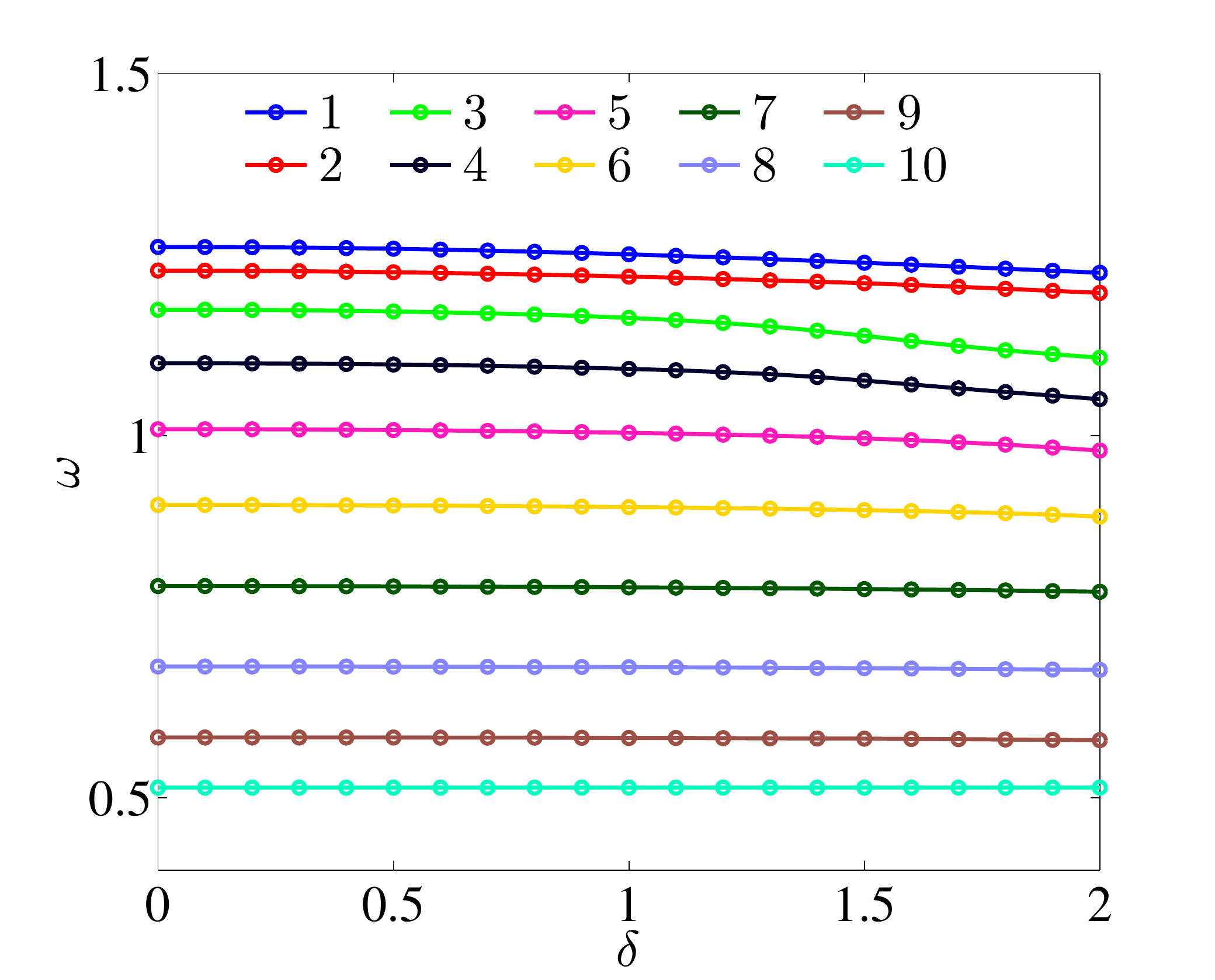}
\includegraphics[width=5cm]{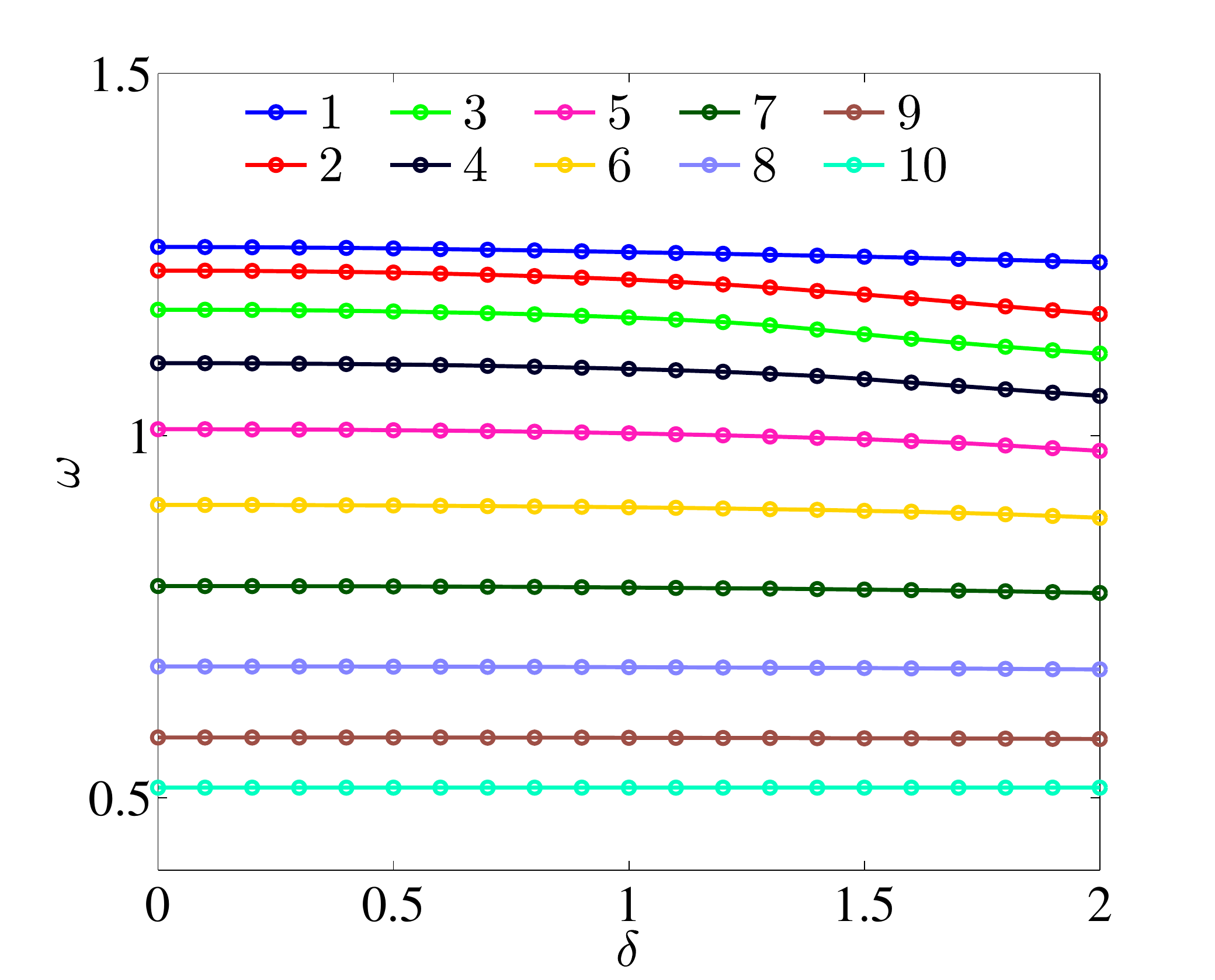}
\includegraphics[width=5cm]{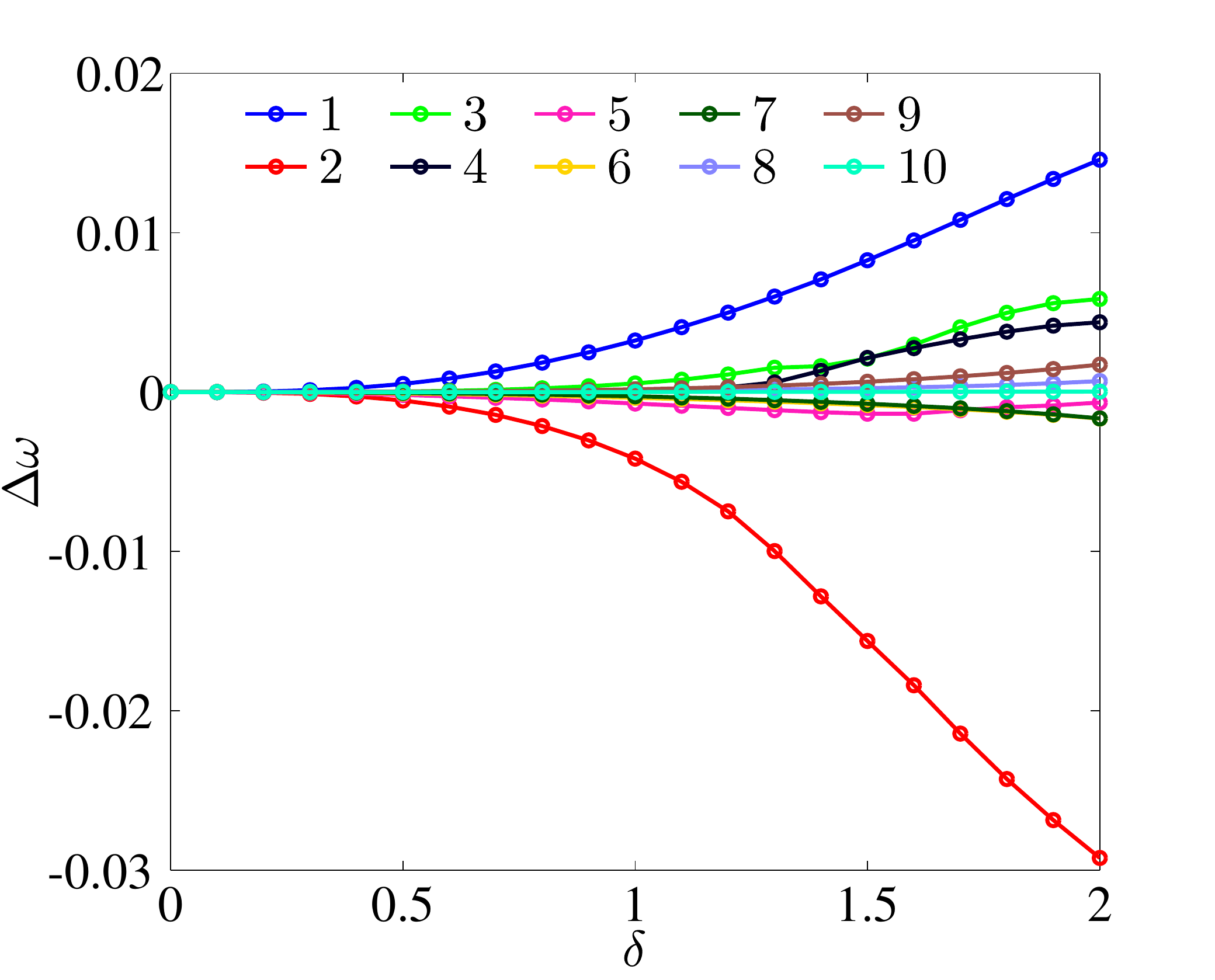}
\end{tabular}
\caption{Here we first numerically find solutions to equation~(\ref{eqn_expand_0}) with different $\omega_j$ and $\delta\in [0,2]$.   Then the left panel shows the change of each frequency $\omega$ in the granular chain (NS-1) (middle panel for (NS-2)) as the amplitude of the solution is increased. We observe that all the frequencies for both systems tend to decrease as $\delta$ increases. In the right panel, we draw the differences between the frequencies of (NS-2) and (NS-1), i.e., $\Delta \omega =\omega|_{(NS-2)}-\omega|_{(NS-1)}$, over the growth of $\delta$. }
\label{fig_omega_numerical}
\end{figure}

\section{Conclusions \& Future Challenges}
In this study, we explored the dynamics of granular chains that are isospectral in their linear limit. Such systems have shown twofold advantages. First, we achieve  considerable freedom in constructing isospectral systems by relying on lumped mass models. Second, we can easily introduce nonlinearity in the system in order to study its effect (such as frequency softening) on the system dynamics. We discussed various spectral transformation schemes in this context, however, we found that the continuous transformation scheme ensures better control over system parameters as opposed to popular factorization techniques such as Cholesky and QR decomposition used in supersymmetric frameworks. We enhanced this spectral transformational scheme by adding the functionality of adding/removing eigenfrequencies of choice. We showed that the entire framework thus allows us to tailor the spectrum details, and construct broad classes of isospectral physical systems, ranging from ordered to seemingly disoredered granular chains.

In terms of the dynamics, we observed that transmission gain is preserved under the spectral transformation scheme we used. For example, an ordered dimer granular chain showed similar transmission characteristics as that of a disordered granular chain isospectral to the former. This indicates new pathways of constructing disordered systems preserving Bragg band-gaps, which are formed due to periodicity in the medium. As the strength of nonlinearity grows, either the amplitude of the solution increases or the precompression decreases, the linear spectrum will be distorted and isospectral systems can lead to quite different dynamics. In particular, a typical scenario is that the natural frequencies of different isospectral granular chains will decrease differently when the amplitude of the solution increases. However, the symmetry of being isospectral may still be of value (and produce quantitatively proximal variations)
for small amplitude excitations.

Given the smooth transitions of the granular system's parameters calculated from the continuous isospectral flow method, we expect to realize the suggested isospectral systems in future experiments. On the theoretical side, 
it is relevant to develop quantitative diagnostics for identifying  how nonlinearity effects lead to a deviation from the identical spectra of linearly isospectral systems. Furthermore, a deeper (and more difficult) question would be whether there exist some analogies or generalizations of these notions for genuinely nonlinear systems, perhaps borrowing ideas from the theory of inverse scattering transforms and associated Lax pair problems. The outcome of these future studies will be reported in the authors' upcoming publications. \\

{\it Acknowledgements.}\\
J.Y. and P.G.K. acknowledge support from US-ARO under grant (W911NF-15-1-0604). \\

\vspace{1cm}
\appendix

\section*{Appendix: Proofs of Remarks and Lemmas}

Proof of Remark \ref{remark_mass_spring_1}:

\begin{proof}
The first two conditions are straightforward to show and we will only prove (A3).

In order to prove the positive definiteness of $H$, it suffices to show $B=G H G$ is positive definite. Thanks to the explicit expression of $B$, we observe that $B$ is a symmetric diagonally dominant real matrix with nonnegative diagonal entries, which implies that $B$ is positive semidefinite. On the other hand, it can be shown by induction that $\det(B)=(\prod_{j=1}^{n+1}K_j)(\sum_{i=1}^{n+1}\frac{1}{K_i})>0$. As a result, $B$ is positive semidefinite and it does not have a
zero eigenvalue thus it must be positive definite.
\end{proof}


Proof of Remark \ref{remark_mass_spring_2}:

\begin{proof}
For a spring-mass system with masses $(\tilde{m}_1,\tilde{m}_2,...,\tilde{m}_n)$ and spring constants $(\tilde{K}_1,\tilde{K}_2,...,\tilde{K}_{n+1})$, its stiffness matrix $\tilde{B}$ satisfies
\begin{eqnarray*}
\tilde{B}\left(
    \begin{array}{c}
        1\\
        1\\
        ...\\
        1\\
        1
    \end{array}
\right)=\left(
    \begin{array}{c}
        \tilde{K}_1\\
        0\\
        ...\\
        0\\
        \tilde{K}_{n+1}
    \end{array}
\right).
\end{eqnarray*}
Let $\tilde{G}=\tilde{M}^{1/2}$ and $\tilde{\boldsymbol{g}}=(\tilde{g}_1, \tilde{g}_2,...,\tilde{g}_n)^T=\tilde{G}(1,1,...,1)^T$ where $\tilde{M}$ is the mass matrix. If the spring-mass system has $\tilde{H}$ as its associated matrix, then $\tilde{B}=\tilde{G}\tilde{H}\tilde{G}$ and
\begin{eqnarray}
\label{linear_system}
\tilde{H}\left(
    \begin{array}{c}
        \tilde{g}_1\\
        \tilde{g}_2\\
        ...\\
        \tilde{g}_{n-1}\\
        \tilde{g}_n
    \end{array}
\right)=\tilde{H}\tilde{G}\left(
    \begin{array}{c}
        1\\
        1\\
        ...\\
        1\\
        1
    \end{array}
\right)=\tilde{G}^{-1}\left(
    \begin{array}{c}
        \tilde{K}_1\\
        0\\
        ...\\
        0\\
        \tilde{K}_{n+1}
    \end{array}
\right)=\left(
    \begin{array}{c}
        \frac{\tilde{K}_1}{\tilde{g}_1}\\
        0\\
        ...\\
        0\\
        \frac{\tilde{K}_{n+1}}{\tilde{g}_n}
    \end{array}
\right).
\end{eqnarray}
Since $\tilde{H}$ is positive definite, there exists a unique $\tilde{\boldsymbol{g}}=\tilde{H}^{-1}(\alpha,0,...,0,\beta)^T$ for any $\alpha>0$ and $\beta>0$. It can be shown that $\tilde{g}$ will be a positive vector as long as $\alpha$ and $\beta$ are positive. Then the masses and the spring constants of the system can be computed from $\tilde{M}=\tilde{G}^2$ and $\tilde{B}=\tilde{G}\tilde{H}\tilde{G}$.

On the positivity of $\tilde{g}$: Since all off-diagonal entries of $\tilde{H}$ are negative and every eigenvalue of $\tilde{H}$ is positive, $\tilde{H}$ is an invertible M-matrix. By the properties of M-matrix, the entries of $\tilde{H}^{-1}$ are all nonnegative. Thus $\tilde{g}=\tilde{H}^{-1}(\alpha,0,...,0,\beta)^T$ will always be a nonnegative vector. In fact, if both $\alpha$ and $\beta$ are positive, we can show $\tilde{g}$ can not have any zero entry due to the fact that all off-diagonal entries of $\tilde{H}$ are negative. By contradiction, if $\tilde{g}_j=0$ for $1<j<n$, then $\tilde{H}(j,j-1)\tilde{g}_{j-1}+\tilde{H}(j,j)\tilde{g}_j+\tilde{H}(j,j+1)\tilde{g}_{j+1}=\tilde{H}(j,j-1)\tilde{g}_{j-1}+\tilde{H}(j,j+1)\tilde{g}_{j+1}=0$ yields $\tilde{g}_{j-1}=\tilde{g}_{j+1}=0$ since $\tilde{g}$ is nonnegative and subdiagonals of $\tilde{H}$ are negative. If $\tilde{g}_1=0$ or $\tilde{g}_n=0$, we have $\alpha=\tilde{H}(1,2)\tilde{g}_2\leq 0 <\alpha$ or $\beta=\tilde{H}(n,n-1)\tilde{g}_{n-1}\leq 0<\beta$, which is impossible. Therefore, we can always obtain a positive solution for $\tilde{g}$ given $\alpha>0$ and $\beta>0$.
\end{proof}


Proof of Lemma \ref{lemma_det}:

\begin{proof}
By direct calculation, we can show $\frac{\partial (\boldsymbol{f}, \boldsymbol{g})}{\partial (\tilde{\boldsymbol{r}}, \tilde{\boldsymbol{\gamma}})}=\left(
    \begin{array}{cc}
        J_{11} & J_{12} \\
        J_{21} & J_{22}
    \end{array}
\right)$ where
\begin{itemize}
\item $J_{11}=\frac{\partial \boldsymbol{f}}{\partial \tilde{\boldsymbol{r}}}$ is an $n\times n$ tridiagonal matrix.
\item $J_{12}=\frac{\partial \boldsymbol{f}}{\partial \tilde{\boldsymbol{\gamma}}}$ is an $n\times (n-1)$ matrix. Its first row is zero while the rest $(n-1)$ rows form a diagonal matrix.
\item $J_{21}=\frac{\partial \boldsymbol{g}}{\partial \tilde{\boldsymbol{r}}}$ is an $(n-1)\times n$ bidiagonal matrix.
\item $J_{22}=\frac{\partial \boldsymbol{g}}{\partial \tilde{\boldsymbol{\gamma}}}=\boldsymbol{0}_{(n-1)\times(n-1)}$.
\end{itemize}
Thus $|\frac{\partial (\boldsymbol{f}, \boldsymbol{g})}{\partial (\tilde{\boldsymbol{r}}, \tilde{\boldsymbol{\gamma}})}|=\det\left(
    \begin{array}{cc}
        L_{11} & J_{12} \\
        J_{21} & J_{22}
    \end{array}
\right)$ where $L_{11}=\left(
    \begin{array}{ccccc}
        \frac{\partial f_1}{\partial \tilde{r}_1} & \frac{\partial f_1}{\partial \tilde{r}_2} & 0 & ... & 0 \\
        0 & 0 & 0 & ... & 0\\
        ... & ... & ... & ... & ...\\
        0 & 0 & 0 & ... & 0
    \end{array}
\right)$.
%
\begin{eqnarray}
\frac{\partial f_1}{\partial \tilde{r}_1}&:=& -\frac{3}{4\pi\rho} [ \frac{1}{3}W \frac{ \tilde{r}_1( \tilde{r}_{2}^{\frac{1}{3}}+8(\tilde{r}_1+\tilde{r}_2)^{\frac{4}{3}} + 8 \tilde{r}_{2}^{\frac{1}{3}} (\tilde{r}_1+\tilde{r}_2)  }{\tilde{r}_1^{\frac{11}{3}}(\tilde{r}_1+\tilde{r}_2)^{\frac{4}{3}}} +\frac{3{\gamma}_1}{\tilde{r}_1^4}] \\
%
\frac{\partial f_j}{\partial \tilde{r}_j}&:=& -\frac{3}{4\pi\rho} [ \frac{1}{3}W \frac{ \tilde{r}_j( \frac{\tilde{r}_{j-1}^{\frac{1}{3}}}{(\tilde{r}_j+\tilde{r}_{j-1})^{\frac{4}{3}}} + \frac{\tilde{r}_{j+1}^{\frac{1}{3}}}{(\tilde{r}_j+\tilde{r}_{j+1})^{\frac{4}{3}}} ) + 8 ( \frac{\tilde{r}_{j-1}^{\frac{1}{3}}}{(\tilde{r}_j+\tilde{r}_{j-1})^{\frac{1}{3}}} + \frac{\tilde{r}_{j+1}^{\frac{1}{3}}}{(\tilde{r}_j+\tilde{r}_{j+1})^{\frac{1}{3}}} ) }{\tilde{r}_j^{\frac{11}{3}}} +\frac{3\tilde{\gamma}_j}{\tilde{r}_j^4}], \quad 2\leq j\leq n-1 \\
\frac{\partial f_n}{\partial \tilde{r}_n}&:=& -\frac{3}{4\pi\rho} [ \frac{1}{3}W \frac{ \tilde{r}_n( \tilde{r}_{n-1}^{\frac{1}{3}}+8(\tilde{r}_{n}+\tilde{r}_{n-1})^{\frac{4}{3}} + 8 \tilde{r}_{n-1}^{\frac{1}{3}} (\tilde{r}_n+\tilde{r}_{n-1}) }{\tilde{r}_{n}^{\frac{11}{3}}(\tilde{r}_n+\tilde{r}_{n-1})^{\frac{4}{3}}} +\frac{3\tilde{\gamma}_n}{\tilde{r}_n^4}] \\
\frac{\partial f_j}{\partial \tilde{r}_{j-1}}&:=& \frac{3}{4\pi\rho}  \frac{1}{3}W \frac{1}{\tilde{r}_{j}^{\frac{5}{3}}{\tilde{r}_{j-1}^{\frac{2}{3}}(\tilde{r}_j+\tilde{r}_{j-1})^{\frac{4}{3}}} }, \quad 2\leq j\leq n\\
\frac{\partial f_j}{\partial \tilde{r}_{j+1}}&:=& \frac{3}{4\pi\rho}  \frac{1}{3}W \frac{1}{\tilde{r}_{j}^{\frac{5}{3}}{\tilde{r}_{j+1}^{\frac{2}{3}}(\tilde{r}_j+\tilde{r}_{j+1})^{\frac{4}{3}}} }, \quad 1\leq j\leq n-1\\
\frac{\partial f_j}{\partial \tilde{\gamma}_{j}}&:=& -\frac{3}{4\pi\rho \tilde{r}_j^3 }, \quad 2\leq j\leq n\\
\frac{\partial g_j}{\partial \tilde{r}_{j}}&:=& \frac{3}{4\pi\rho}  \frac{1}{3}W \frac{9\tilde{r}_j+7\tilde{r}_{j+1}}{6\tilde{r}_{j}^{\frac{13}{6}}{\tilde{r}_{j+1}^{\frac{7}{6}}(\tilde{r}_j+\tilde{r}_{j+1})^{\frac{4}{3}}} }, \quad 1\leq j\leq n-1\\
\frac{\partial g_j}{\partial \tilde{r}_{j+1}}&:=& \frac{3}{4\pi\rho}  \frac{1}{3}W \frac{7\tilde{r}_j+9\tilde{r}_{j+1}}{6\tilde{r}_{j+1}^{\frac{13}{6}}{\tilde{r}_{j}^{\frac{7}{6}}(\tilde{r}_j+\tilde{r}_{j+1})^{\frac{4}{3}}} }, \quad 1\leq j\leq n-1\\
\end{eqnarray}
If $(\tilde{r}_1,\tilde{r}_2,...,\tilde{r}_n)$ are positive and $\gamma_1\geq 0$, then $\frac{\partial f_1}{\partial \tilde{r}_1}<0$, $\frac{\partial f_1}{\partial \tilde{r}_{2}}>0$, $\frac{\partial f_j}{\partial \tilde{\gamma}_j}<0$ for $2\leq j\leq n$, $\frac{\partial g_j}{\partial \tilde{r}_j}>0$ and $\frac{\partial g_j}{\partial \tilde{r}_{j+1}}>0$ for $1\leq j\leq n-1$. Since $\det\left(
    \begin{array}{cc}
        \frac{\partial f_1}{\partial \tilde{r}_1} & \frac{\partial f_1}{\partial \tilde{r}_2}\\
        \frac{\partial g_1}{\partial \tilde{r}_1} & \frac{\partial g_1}{\partial \tilde{r}_2}
    \end{array}
\right)<0$, it can be checked that $|\frac{\partial (\boldsymbol{f}, \boldsymbol{g})}{\partial (\tilde{\boldsymbol{r}}, \tilde{\boldsymbol{\gamma}})}|=\det\left(
    \begin{array}{cc}
        L_{11} & J_{12} \\
        J_{21} & J_{22}
    \end{array}
\right)\neq 0$.
\end{proof}

\vspace{1cm}

\end{document}
%